\documentclass[letterpaper,twocolumn,english,preprintnumbers,amsmath,amssymb,superscriptaddress,footinbib,prx]{revtex4-2} 
\usepackage{graphicx}

 \usepackage{url}

\graphicspath{{./Images/}}
\usepackage{xcolor}
\usepackage[colorlinks = true, citecolor=blue, linkcolor=black,urlcolor=purple]{hyperref}
\usepackage{enumitem}
    \setlist[itemize]{noitemsep, topsep=0pt}
    \setlist[enumerate]{noitemsep, topsep=0pt}
\usepackage{verbatim}
\usepackage{mathtools,amssymb,amsmath}
\usepackage{bm}
\usepackage{yhmath}
\usepackage[inline]{asymptote}
\usepackage{cancel}
\usepackage{relsize}
\usepackage{array}
\usepackage{bm}
\newtheorem{thm}{Theorem}

\newcommand{\bra}[1]{\langle#1|}
\newcommand{\ket}[1]{|#1\rangle}

\newcommand{\be}{\begin{equation}}
\newcommand{\ee}{\end{equation}}
\newcommand{\<}{\langle}
\renewcommand{\>}{\rangle}

\newcommand{\Tr}{{\rm Tr\,}}

\usepackage[normalem]{ulem}
\usepackage{color}

\usepackage{algorithm}
\usepackage{algpseudocode}

\renewcommand{\vec}[1]{{\bf #1}}
\newtheorem{prop}{Proposition}
\newenvironment{proof}{\textbf{Proof}}{\hfill$\square$}

\newtheorem{cor}{Corollary}
\newtheorem{lem}{Lemma}

\newtheorem{rem}{Remark}

\newtheorem{clm}{Claim}

\usepackage{tikz}
\usepackage{xcolor}
\usetikzlibrary{patterns}
\allowdisplaybreaks

\usepackage{comment}

\newcommand{\almaden}{IBM Quantum, IBM Research -- Almaden, San Jose CA, 95120, USA}
\newcommand{\cambridge}{IBM Quantum, MIT-IBM Watson AI lab,  Cambridge MA, 02142, USA}
\newcommand{\mitaff}{Department of Physics, Massachusetts Institute of Technology, Cambridge, MA, 02139, USA}
 \newcommand{\iasaff}{School of Natural Sciences, Institute for Advanced Study, Princeton, NJ, 08540, USA}
\newcommand{\google}{Google Quantum AI, Venice Beach, CA, 90291, USA}

\newcommand{\ourtitle}{Advantage of Quantum Neural Networks as Quantum Information Decoders}

\begin{document}

\title{Advantage of Quantum Neural Networks as Quantum Information Decoders}

\author{Weishun Zhong}
\email{wszhong@ias.edu}
\affiliation{Department of Physics, Massachusetts Institute of Technology, Cambridge, MA, 02139, USA}
\affiliation{IBM Quantum, MIT-IBM Watson AI lab, Cambridge MA, 02142, USA}
 \affiliation{School of Natural Sciences, Institute for Advanced Study, Princeton, NJ, 08540, USA}
\color{black}

\author{Oles Shtanko}
\email{oles.shtanko@ibm.com}
\affiliation{IBM Quantum, IBM Research, Almaden, San Jose CA, 95120, USA}

\author{Ramis Movassagh}
\email{movassagh@google.com}
\affiliation{IBM Quantum, MIT-IBM Watson AI lab, Cambridge MA, 02142, USA}\affiliation{Google Quantum AI, Venice, CA, 90291, USA}\

\begin{abstract}
A promising strategy to protect quantum information from noise-induced errors is to encode it into the low-energy states of a topological quantum memory device. 
However, readout errors from such memory under realistic settings is less understood. We study the problem of decoding quantum information encoded in the groundspaces of topological stabilizer Hamiltonians in the presence of generic perturbations, such as quenched disorder. 
We first prove that the standard stabilizer-based error correction and decoding schemes work adequately well in such perturbed quantum codes by showing that the decoding error diminishes exponentially in the distance of the underlying unperturbed code. We then prove that  Quantum Neural Network (QNN) decoders provide an almost quadratic improvement on the readout error.
Thus, we demonstrate provable advantage of using QNNs for decoding realistic quantum error-correcting codes, and our result enables the exploration of a wider range of non-stabilizer codes in the near-term laboratory settings.
\end{abstract}
\maketitle

Many physical quantum systems, including superconducting qubits \cite{bravyi2022future}, trapped ions \cite{bruzewicz2019trapped}, and cold atoms \cite{gross2017quantum}, are widely regarded as promising platforms for quantum computers. However, they are susceptible to noise, which prevents them from reliably storing and processing quantum information \cite{preskill2018quantum,bharti2021noisy}. One promising strategy to overcome these is to perform Quantum Error Correction (QEC) \cite{bravyi2022future}, which involves encoding information using a set of noise-insensitive states. Physically, quantum information can be encoded in the low-energy states of systems described by topological Hamiltonians \cite{kitaev2003fault,bravyi1998quantum,fowler2012surface,dennis2002topological,raussendorf2007fault,satzinger2021realizing,semeghini2021probing}. One of the most common examples of topological Hamiltonians are the stabilizer Hamiltonians, which are the sum of the parity check of a stabilizer code \cite{shor1995scheme, calderbank1996good, steane1996error, calderbank1998quantum, gottesman1997stabilizer}. 
However, stabilizer Hamiltonians are finely tuned. Thus, here we embark on an exploration of a more general case of Hamiltonians that are not an ideal stabilizer Hamiltonians, but are subject to local perturbations. We examine the limitations of quantum error correction under such conditions. We then propose an alternative, Hamiltonian-agnostic approach to decoding information for such codes based on quantum machine learning.

In the last decade, we have witnessed the spectacular success of machine learning with noisy classical data in a variety of tasks, such as classification, regression, and generative modeling \cite{vincent2010stacked,krizhevsky2017imagenet,kingma2013auto,goodfellow2016deep,goodfellow2020generative}. Recently, there has been a considerable amount of research devoted to the application of machine learning to the study of quantum many-body systems \cite{carleo2017solving,biamonte2017quantum,van2017learning,carrasquilla2017machine,carleo2019machine,wang2016discovering,zhang2017quantum,biamonte2017quantum,cong2019quantum, han2018unsupervised,gomez2022reconstructing,zhong2022many}, where the complexity of many-body interactions makes analytical or numerical treatments formidable. In particular, it has been demonstrated that parameterized quantum circuits \cite{benedetti2019parameterized,bharti2021noisy,abbas2021power}, also called Quantum Neural Networks (QNNs), can recognize quantum states in different phases and perform QEC \cite{cong2019quantum,locher2023quantum,cao2022quantum}. Despite its success in small-scale tasks, training quantum neural networks is known to suffer from the barren plateau problem, where gradients become exponentially small in the system size, and optimization resembles searching for a needle in a haystack \cite{mcclean2018barren}. Nevertheless, it has been recently demonstrated that the barren plateau problem can be circumvented by utilizing logarithmic-depth convolutional networks with local output observables \cite{cerezo2021cost}. Here, we propose to employ QNNs for decoding imperfect stabilizer codes that are low-energy states of topological Hamiltonians, and show their advantage over standard QEC.

In this work we introduce and study stabilizer Hamiltonians \cite{movassagh2020constructing} perturbed by a generic local quenched disorder. We establish rigorous results on decoding errors using standard QEC methods. In particular, we show that simple measurements of logical operators are not sufficient to achieve fault tolerance, and that a sequence of stabilizer measurements followed by error corrections are required to decode even noiseless input states. 
We show theoretically that QNNs can outperform standard QEC methods. We also study numerically how to construct such QNNs for minimal circuit architectures. \\
\begin{figure*}[ht]
    \centering
        \includegraphics[width=1.01\textwidth]{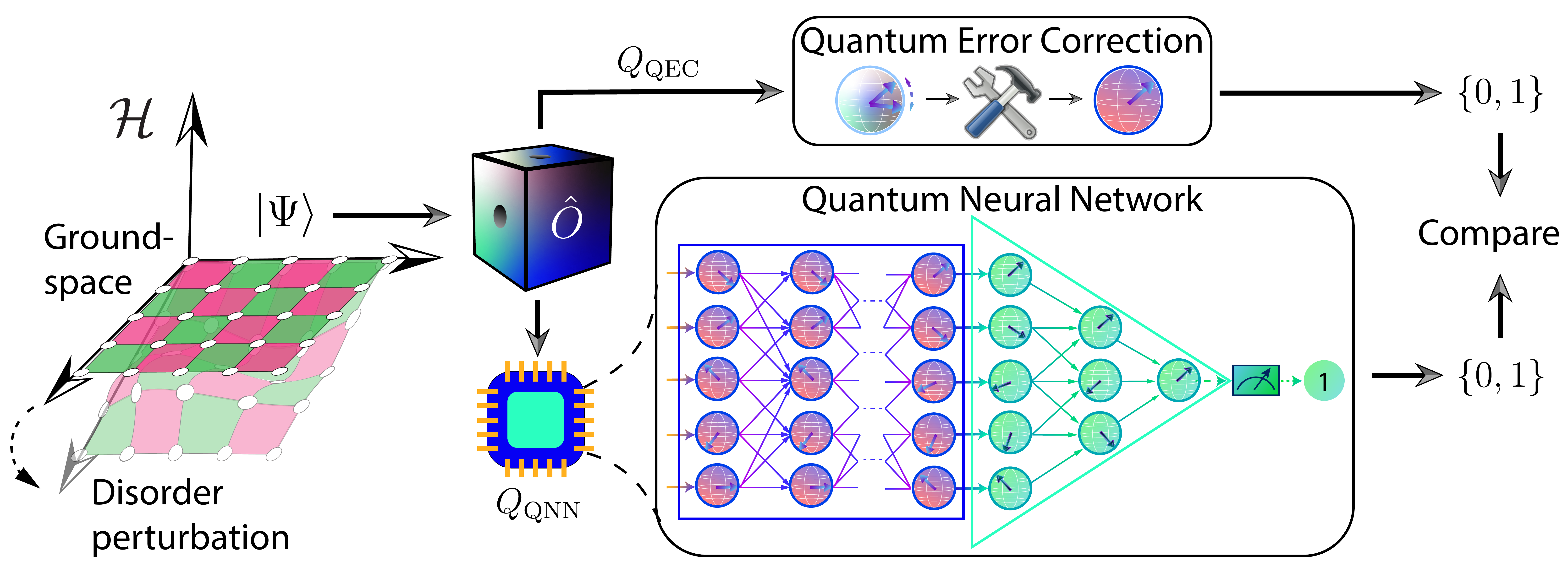}
        \caption{\textbf{Schematics of decoding realistic quantum codes.} From left to right: Stabilizer codes realized as ground states of physical Hamiltonians get perturbed and yield imperfect codewords. A state $\ket{\Psi}$ prepared by the imperfect codewords is presented to the magic box decoder $\hat{O}$ for subsequent decoding. (i) Top row: $\hat{O}=Q_{\text{QEC}}$. $\ket{\Psi}$ is decoded through standard QEC procedure. (ii) Bottom row: $\hat{O}=Q_{\text{QNN}}$. $\ket{\Psi}$ is decoded through a quantum neural network. After decoding, bit strings are sampled from the resulting probability distributions and the outcomes of the two decoding procedures are compared. }
\label{fig:schematics}
\end{figure*}

\textbf{Problem setup.}
We study a model of a quantum code spanning the groundspace of a generic topological Hamiltonian. We start our analysis with a general stabilizer quantum code $[[n,1,d]]$ of distance $d$, which encodes a single logical qubit using $n$ physical qubits. We define $\{S_a\}$ as a set of $n-1$ stabilizer operators with the condition that $[S_a, S_b] = 0$. Furthermore, we consider $X_L$ and $Z_L$ to be the logical Pauli-X and Pauli-Z operators of this stabilizer code, respectively. Next, we introduce $V$ as a generic $k$-local perturbation ($k<d$), which is a sum of spatially local terms acting on at most $k$ qubits. Then, we define the Hamiltonian as
\begin{align}
\label{eqn:hamiltonian}
    H =  H_0 + \lambda V, \qquad H_0 = -\sum_{a=1}^{n-1} S_a,
\end{align}
where $H_0$ is the non-perturbed stabilizer Hamiltonian, $V$ is a $k$-local perturbation defined above, and $\lambda$ is a small real parameter that sets the perturbation strength. 

When $\lambda=0$ , the groundspace of $H$ is doubly degenerate, spanned by the codewords of the original stabilizer code. 
If $0<\lambda< \Delta$, where $\Delta$ is the gap of the Hamiltonian $H_0$, the two near-degenerate ground states remain an approximate quantum code.
This means we can encode information using the two lowest eigenstates $\ket{\psi_0}$ and $\ket{\psi_1}$ of the perturbed Hamiltonian, $H \ket{\psi_i} = E_i \ket{\psi_i}$.
The energy splitting $\delta = |E_1-E_0|$ is exponentially small in the stabilizer code's distance $d$ \cite{bravyi2010topological}. Specifically, we encode the quantum information for a single logical qubit using quantum states of the form
\begin{align}
\label{eqn:psi}
    \ket{\Psi} = c_0 \ket{\Omega_0} + c_1 \ket{\Omega_1}   ,
\end{align}
where $|\Omega_i\> = \sum_j u_{ij}|\psi_j\>$ are the codewords, $u$ is a Hamiltonian-dependent $2\times 2$ unitary matrix, 
$c_0=\cos\theta$ and $c_1=e^{i\phi}\sin\theta$ represent the amplitudes of the encoded logical state on the Bloch sphere, with $\theta\in[0,\pi]$ and $\phi \in [0,2\pi]$ being the coordinates. There are two distinct sources of error that affect the precision of decoding: the perturbation $V$ that makes the codewords imperfect, and the potential noise at the input. We model such noise by stochastically applying $p$-local Pauli operators $P_\alpha$ to input states in Eq.~\eqref{eqn:psi}. The integer parameter $p>0$, which describes the locality of the noise error, serves here as a measure of the noise strength. 

\begin{figure*}[t!]
    \centering
        \includegraphics[width=1.0\textwidth]{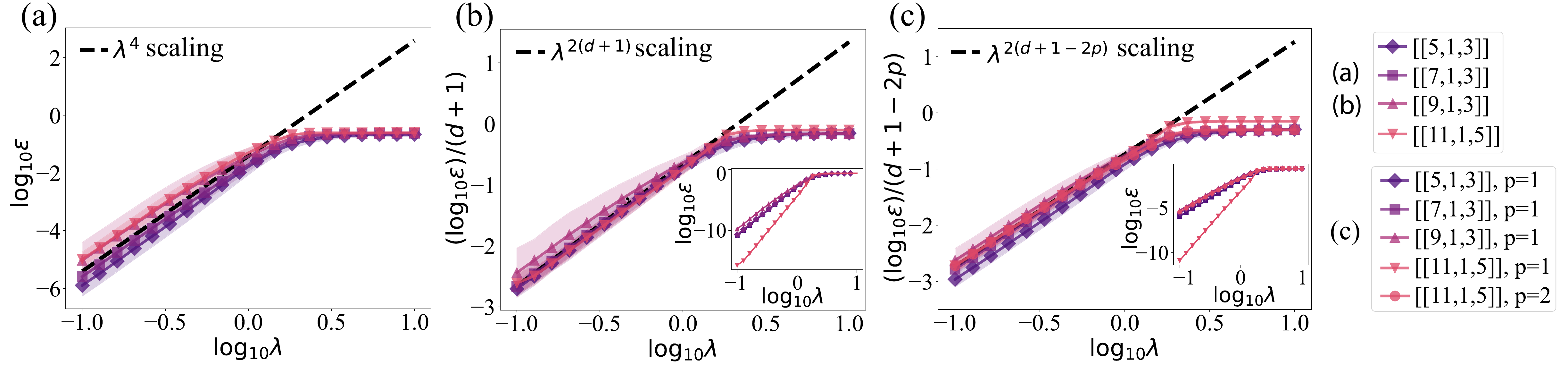}
        \caption{\textbf{Universality of standard decoding.} Performance of standard decoding strategies on perturbed stabilizer codes. On $y$-axis we abbreviate the generalization error in $X$ basis $\varepsilon_X$ by $\varepsilon$. The universal scaling laws  $\varepsilon\sim\lambda^\alpha$ are shown in dashed lines. Panel (a) show the error of measuring logical operator $X_L$ as a function of perturbation strength $\lambda$. Dashed line has scaling exponent $\alpha=4$, as predicted by Theorem~\ref{thm1}. Panel (b) shows the rescaled logarithm of the error for measuring  logical operator $X_L$ \textit{after} measuring syndromes and applying error correction. Dashed line has scaling exponent $\alpha = 2(d+1)$, consistent with Theorem~\ref{thm2}. Panel (c) is the same as panel (b) but for noisy input states. Dashed line has scaling exponent $\alpha = 2(d+1-2p)$, as predicted by Theorem~\ref{thm2}. Note that all the curves from different codes nearly collapse into a single, universal shape in (a)-(c) after dividing by the scaling exponent of the dashed line. Inset in (b)-(c) shows the raw unrescaled data which have different slopes depending on the code distance.}
\label{fig:QEC_panel}
\end{figure*}

Our goal is to construct a logical state \textit{decoder} that takes $|\Psi\>$ as input and returns either $0$ or $1$ after it has measured the qubit on logical basis $Q \in \{X,Y,Z\}$. 
We illustrate this decoding process in Fig.~\ref{fig:schematics}. Specifically, for the $Z$ decoder, the expectation of the qubit output is $f_Z(\theta,\phi) :=c_0^2-|c_1|^2 =  \cos2\theta$. Similarly, the $X$ and $Y$ decoders should return the expectations $f_X(\theta,\phi) = \sin2\theta\cos\phi$ and $f_Y(\theta,\phi) = \sin2\theta\sin\phi$, respectively. 

From a mathematical point of view, the decoding is a general measurement of a certain (non-local) operator $\hat{O}$. To evaluate the \textit{generalization error} of the decoding $\ket{\Psi}$, we use the mean square error as a metric to measure the deviation of the prediction from the ground truth,
\begin{align}
\label{eqn:L2_logical}
    \varepsilon_Q (\hat{O})  =  \mathbb E_\alpha \int d\theta d\phi \,\mu(\theta,\phi)  \Bigl(\<\Psi|P_\alpha \hat{O} P_\alpha|\Psi\>  -  f_Q(\theta,\phi)\Bigl)^2,
\end{align}
where $\mu(\phi,\theta)$ is a measure that depends on the input distribution, $\mathbb E_\alpha$ is expectation over input Pauli errors $P_\alpha$, and $f_Q(\theta,\phi)$ is the function that describes the desired output expectation value for the measurement in the basis $Q$.

\textbf{Limitations of standard decoding}. The conventional method of constructing a decoder depends on the information available about the system. First, we assume that we only know the logical operators $Q_L\in\{X_L,Y_L,Z_L\}$ for the unperturbed stabilizer code. Given this information, the only accessible action is to measure $Q_L$ directly. This approach could be used, for instance, in an experimental setting \cite{semeghini2021probing} where it is known that the logical operators must be a string of Pauli operators, but no efficient decoder is available. The error associated with this approach is given by the following theorem.

\begin{thm}
\label{thm1}(naive decoding, informal)
 For noiseless input and generic
 $V$, the generalization error of measuring the logical operator is $\varepsilon_Q(Q_L) = \Theta(\lambda^{4})$.
\end{thm}

Here we use the $\Theta$ (big-Theta) asymptotic notation to indicate the exact scaling of the error
with perturbation strength $\lambda$, while holding everything else fixed (number of qubits, distance, etc.). The formal version of the theorem along with the proof can be found in Supplementary Information (SI) Section~\ref{app:proof_thm1}.

We illustrate the result of Theorem 1 by numerically simulating the result of measurement for three different stabilizer codes of distance $d=3$: the five-qubit code \cite{laflamme1996perfect}, the Steane code \cite{steane1996error}, and Shor's code \cite{shor1995scheme}. We also consider the smallest distance-5 code, $[[11,1,5]]$ \cite{Grassl:codetables}. As for perturbation, in all cases $V = \sum_{i=1}^n V_i$, where $V_i$ are single-qubit matrices sampled from the Gaussian unitary ensemble and normalized so that its spectral norm is unity, $||V_i|| = 1$. The results are shown in Fig.~\ref{fig:QEC_panel}(a). 

\begin{figure*}[t!]
    \centering
        \includegraphics[width=1.01\textwidth]{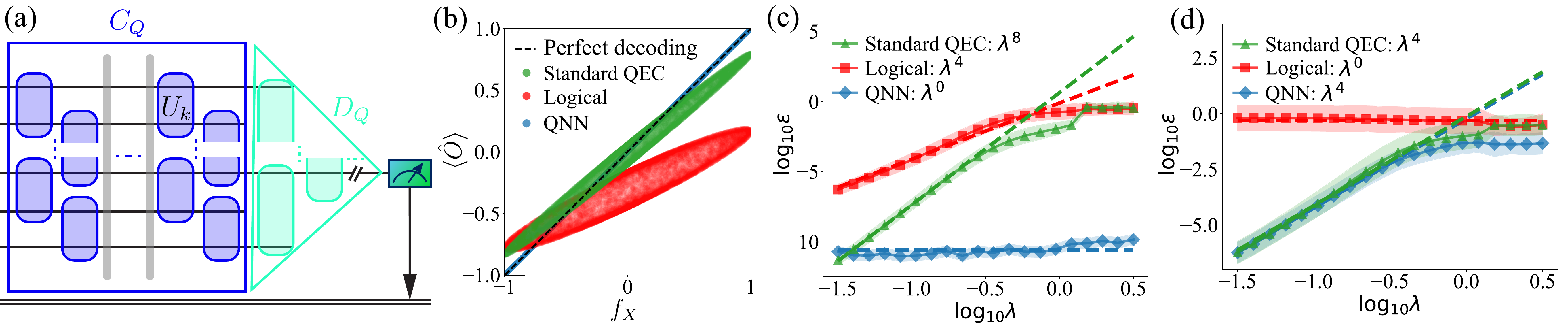}
        \caption{\textbf{QNN performance for 5-qubit code.} (a) Our QNN architecture consists of an error-correction circuit $C_Q$ and a decoding circuit $D_Q$. (b) Decoding noiseless states for perturbed 5-qubit code  at $\lambda=1$. (c)-(d) Example decoding of imperfect 5-qubit code using $X_L$ (red), $X_{\text{QEC}}$ (green) and $X_{\text{QNN}}$ (blue). (c) Decoding noiseless states. (d) Decoding noisy states corrupted by a randomly chosen single-qubit Pauli operator $P_\alpha \in \{I, X, Y, Z\}$ acting on a randomly chosen qubit. Panels (b)-(c) demonstrate that QNN significantly outperforms the standard approaches across nearly two orders of magnitude in $\lambda$. Panel (d) confirms that the scaling of standard QEC and QNN agree with our analytical predictions in Theorems \ref{thm2}-\ref{thm3}.}
\label{fig:QNN_panel}
\end{figure*}

Theorem~\ref{thm1} shows that the error vanishes as $\lambda^4$, independent of the distance $d$. Therefore, measuring logical operators for the unperturbed code does not help if the goal is to achieve fault tolerance by increasing the code distance with more physical qubits.
Moreover, simply measuring logical operators for noisy inputs ($P_\alpha \neq I$) will necessarily lead to a non-vanishing error for any $\lambda$, including $\lambda=0$. 

The decoding quality can be significantly improved if in additional to the logical operators $Q_L$, one also knows the stabilizer operators $S_a$ from the unperturbed Hamiltonian (Eq.~\eqref{eqn:hamiltonian}). In this scenario, one could measure the syndromes of the input state and then apply error correction based on the measured syndromes. This procedure with high probability transforms states from the codespace of the perturbed Hamiltonian into the codespace of original stabilizer code. The logical operators can then be measured to obtain the result. The combination of these two procedures is equivalent to measuring the operator
\begin{align}
    \label{eqn:logical_tilde}
    Q_{\text{QEC}}:= \mathcal{E}_{\rm QEC}^{\dagger} (Q_L),
\end{align}
where $\mathcal{E}_{\rm QEC}(\cdot)$ is the quantum error-correction map (see SI Section \ref{app:proof_thm2}), and $\mathcal A^\dag$ denotes the adjacent map for a superoperator $\mathcal A$, i.e., $\Tr(O_1\mathcal A(O_2)) = \Tr(\mathcal A^\dag(O_1) O_2)$. The associated error is given by the following theorem (the formal version, as well as the proof can be found in SI Section~\ref{app:proof_thm2}).
\begin{thm}(decoding with error correction, informal)
\label{thm2}
For noiseless input, the generalization error for standard QEC is $\varepsilon_Q(Q_{\rm QEC}) = O(\lambda^{2\lceil d/k\rceil})$. For noisy input ($p>1$) and generic $V$, the generalization error is $\varepsilon_Q(Q_{\rm QEC}) = \Theta(\lambda^{2\lceil(d+1-2p)/k\rceil})$. 
\end{thm}
Here we use $O$ notation to indicate asymptotic upper bound as $\lambda \to 0$ and $\lceil x\rceil$ is a ceiling function that rounds up any non-integer $x$. As a corollary, for states unaffected by noise ($p=0$), turning on error correction restores dependence on the code distance in the error scaling. In particular, for 1-local perturbations, the $\Theta(\lambda^4)$ constant scaling of Theorem~\ref{thm1} becomes $ O(\lambda^{2d})$.

We illustrate this result in Fig.~\ref{fig:QEC_panel}(b) and (c), where we present numerical simulations of the codes and the perturbation $V$ considered earlier. For the noiseless case  (Fig.~\ref{fig:QEC_panel}(b)), we find that for the codes we study, all curves have the same $\lambda^{2(d+1)}$ scaling, consistent with the $O(\lambda^{2d})$ upper bound in Theorem~\ref{thm2}. To illustrate this feature, we show the collapsed version of the original data (shown in the inset), $\log(\varepsilon_Q (Q_{\text{QEC}}))/(d+1)$ versus $\lambda$. The case where the state $\ket{\Psi}$ is affected by $p$ one-qubit Pauli errors is shown in Fig.~\ref{fig:QEC_panel}(c). Similar to the noiseless case, we show that the values of $\log(\varepsilon_Q (Q_{\text{QEC}}))/(d+1-2p)$ collapse to the same slope, agrees with our tight bound in Theorem~\ref{thm2}.

The standard decoding approaches considered in Theorem \ref{thm1}-\ref{thm2} cannot be applied if one has no knowledge about the unperturbed code. In such scenario, variational methods such as the QNN can be utilized. In the following, we do {\it not} assume any knowledge of the perturbed Hamiltonian $H$ beyond the access to the two spanning codewords of its ground space. Nevertheless, as we show below, QNNs can outperform the standard QEC.

\textbf{Decoding using QNN}. Motivated by recent  machine learning  enabled studies of quantum many-body systems~\cite{cong2019quantum, huang2021provably,huang2022quantum}, we propose to use QNNs for decoding non-stabilizer codes. Essentially, QNN is a unitary circuit that takes the encoded state as input and returns the result by measuring the output qubit. If the QNN transformation is represented by the unitary operator $U_Q$ and $Z_0$ is the Pauli-$Z$ operator of the output qubit, the whole procedure is equivalent to measuring the operator
\begin{equation}
\label{eqn:Xbar}
     Q_{\rm QNN} :=  U_Q^{\dagger}Z_0U_Q.
\end{equation}
In the absence of noise in the input, QNN decoding, unlike standard QEC, can achieve arbitrary precision at sufficiently large depths. This is because, for arbitrary pair of codewords, there exists a unitary \(U_Q\) that transforms them into a product state with different states of the target qubit.
 In the presence of noise, QNN provides an almost quadratic error improvement over the standard QEC result in Theorem~2, as detailed in the following theorem.
\begin{thm} \label{thm:qnn_dec}
\label{thm3} (QNN decoding, informal)
For noisy inputs, one could always find a QNN such that its generalization error for decoding is $\varepsilon_Q(Q_{\rm QNN}) = O(\lambda^{4\lceil(d-2p)/k\rceil})$.
\end{thm}
The formal version of the theorem and its proof can be found in SI Section~\ref{app:proof_thm3}.
It is also possible to show that the error scaling with $\lambda$ in Theorem~\ref{thm3} cannot be improved in general. Indeed we show that this bound is saturated in the following proposition.
\begin{prop}\label{eq:prop_impr}
\label{prop1} (informal)
Consider noisy inputs with distribution $\mu(\theta,\phi) = \mu(\theta,\phi+\pi)$ and label values $f_Q(\theta,\phi) = \pm1$. Then for generic $V$, all QNN unitaries $U_Q$ must satisfy $\varepsilon_Q(Q_{\rm QNN}) = \Omega(\lambda^{4\lceil(d-2p)/k\rceil})$.
\end{prop}
Here we use $\Omega$ notation to indicate asymptotic lower bound as $\lambda \to 0$. For a formal version of Proposition~\ref{eq:prop_impr} with accompanying proof, see SI Section~\ref{app:optimal_theorem3}. This rigorous result, as well as Theorem~\ref{thm:qnn_dec}, were proved without any restrictions on the unitary
$U_Q$. In practice, achieving this optimal QNN may involve numerous gates and a complex optimization process. In what follows we numerically illustrate our results in the context of an exampled QNN.

Our network, trainable to sample in logical $Q$-basis, consists of two blocks shown in Fig.~\ref{fig:QNN_panel}(a): correction circuit $C_Q$ (blue) and decoding circuit $D_Q$ (green). The correction circuit is composed of local two-qubit gates forming the brickwork pattern of depth $d_C$.
In turn, the decoding circuit is a convolutional network consisting of $O(\log n)$ layers of two-qubit unitaries, discarding half the qubits after each layer, similar to previous proposals \cite{cong2019quantum,cerezo2021cost}. If the combined network has an overall logarithmic depth, we expect that it does not suffer from the barren plateau problem \cite{mcclean2018barren} and can be efficiently optimized.
We parametrize all the 2-qubit unitaries $U_{k}$ composing the circuit as 
\begin{equation}
\label{eqn:log-network}
     U_{k} = \exp \left\{ i \sum_{\alpha,\beta=1}^4 \varphi^k_{\alpha \beta} \sigma_{\alpha}\otimes \sigma_{\beta} \right\},
\end{equation}
where $\varphi^k_{\alpha \beta}\in[0,2\pi]$ are trainable angles and $\sigma_\alpha$ are single-qubit Pauli matrices.

We assume access to a limited set of training and validation states in code space $\{|\Psi_\mu\>\}$, with equal probabilities affected by $p$-local Pauli errors, and the corresponding labels $(\theta_\mu,\phi_\mu)$. In practice, one can assume such access, as the groundspace can be prepared adiabatically in experimental settings. We optimize the QNN with objective function $\varepsilon_Q(Q_{\rm QNN})$ (Eq.~\eqref{eqn:L2_logical}) on the training data. After training, we estimate the QNN decoding error from the validation data withheld from training. 

We provide numerical simulation results in Fig.~\ref{fig:QNN_panel}, where we study a uniform 1-local perturbation $V_i = \sum_{i=1}^n(X_i+Z_i)$. Fig.~\ref{fig:QNN_panel}(b)-(c) suggest that our QNN architecture ($C=4$) can decode noiseless inputs and achieves superior performance compared to standard decoding schemes discussed in Theorems \ref{thm1}-\ref{thm2}, regardless of the perturbation strength. The QNN error can be further reduced by considering deeper architectures at the cost of longer training times. 

The noisy scenario is shown in Fig.~\ref{fig:QNN_panel}(d). Incidentally, for the 5-qubit code ($d=3,p=1$) simulation, the lower bound in Theorem \ref{thm3} predicts the same $\lambda^4$ scaling as in Theorem \ref{thm2}, which the present QNN saturates. However, with QNN  one can potentially achieve better scaling by considering larger code distances such as the [[11,1,5]] code. Intriguingly, while QNN performance is comparable to that of standard QEC for small $\lambda$, and becomes superior for large $\lambda$. Meanwhile, this approach does not require any knowledge of the Hamiltonian, as needed in standard QEC methods. Also, this result illustrates that sequential error correction and measurement of the logical operator at the end of the circuit can be implemented by applying a QNN unitary and single-qubit measurement without using any ancilla. \\

\textbf{Discussion}. In this work, we have developed a theoretical framework for decoding stabilizer codes perturbed by disorder. We proved general performance bounds regarding the decoding errors of standard QEC versus QNN decoding techniques for such codes. Our results reveal universal scaling behaviors of the decoding error for arbitrary code distances, and suggests a provable advantage of using QNNs over standard QEC methods for decoding realistic codes with large code distances. Our analytical work was confirmed and corroborated  by numerical simulations of codes with various distances ranging from 5 to 11 qubits. 
The numerical methodology may find use-cases elsewhere. 

An interesting future direction is to find whether it is possible to construct logarithmic-depth QNNs such as convolutional NNs that saturate the bound in Theorem~\ref{thm3}. Also, our study reveals a minor gap between the $ O(\lambda^{2 d})$ upper bound scaling in Theorem~\ref{thm2} and the empirical $ O(\lambda^{2(d+1)})$ performance for standard quantum error correction with noiseless input states for $k=1$ local perturbation. Whether the  bound  in Theorem~\ref{thm2} can be improved remains an open question.

Our theoretical framework and numerical protocol represent a first step toward decoding non-stabilizer imperfect codes, which may be of practical relevance in the near-term quantum devices. Our work paves the way for novel applications of machine learning in quantum error correction. While we focus on imperfect stabilizer codes, an immediate question is whether our formalism can be extended beyond near-stabilizer codes. In addition, how does the structure and depth of QNNs affect practical decoding performance remains an important open problem. \\

\acknowledgments{
\textbf{Acknowledgments}. W.Z. acknowledges support from the IBM Quantum Summer Internship Program, where this work initiated. W.Z. also acknowledges support from the Starr Foundation at the Institute for Advanced Study, where the final phase of the project was completed. O.S. and R.M. acknowledge funding from the MIT-IBM Watson AI Lab under the project Machine Learning in Hilbert space. The research was partly supported by the IBM Research Frontiers Institute.}

\bibliography{ref.bib}

\begin{thebibliography}{48}%
\makeatletter
\providecommand \@ifxundefined [1]{%
 \@ifx{#1\undefined}
}%
\providecommand \@ifnum [1]{%
 \ifnum #1\expandafter \@firstoftwo
 \else \expandafter \@secondoftwo
 \fi
}%
\providecommand \@ifx [1]{%
 \ifx #1\expandafter \@firstoftwo
 \else \expandafter \@secondoftwo
 \fi
}%
\providecommand \natexlab [1]{#1}%
\providecommand \enquote  [1]{``#1''}%
\providecommand \bibnamefont  [1]{#1}%
\providecommand \bibfnamefont [1]{#1}%
\providecommand \citenamefont [1]{#1}%
\providecommand \href@noop [0]{\@secondoftwo}%
\providecommand \href [0]{\begingroup \@sanitize@url \@href}%
\providecommand \@href[1]{\@@startlink{#1}\@@href}%
\providecommand \@@href[1]{\endgroup#1\@@endlink}%
\providecommand \@sanitize@url [0]{\catcode `\\12\catcode `\$12\catcode
  `\&12\catcode `\#12\catcode `\^12\catcode `\_12\catcode `\%12\relax}%
\providecommand \@@startlink[1]{}%
\providecommand \@@endlink[0]{}%
\providecommand \url  [0]{\begingroup\@sanitize@url \@url }%
\providecommand \@url [1]{\endgroup\@href {#1}{\urlprefix }}%
\providecommand \urlprefix  [0]{URL }%
\providecommand \Eprint [0]{\href }%
\providecommand \doibase [0]{https://doi.org/}%
\providecommand \selectlanguage [0]{\@gobble}%
\providecommand \bibinfo  [0]{\@secondoftwo}%
\providecommand \bibfield  [0]{\@secondoftwo}%
\providecommand \translation [1]{[#1]}%
\providecommand \BibitemOpen [0]{}%
\providecommand \bibitemStop [0]{}%
\providecommand \bibitemNoStop [0]{.\EOS\space}%
\providecommand \EOS [0]{\spacefactor3000\relax}%
\providecommand \BibitemShut  [1]{\csname bibitem#1\endcsname}%
\let\auto@bib@innerbib\@empty
\bibitem [{\citenamefont {Bravyi}\ \emph {et~al.}(2022)\citenamefont {Bravyi},
  \citenamefont {Dial}, \citenamefont {Gambetta}, \citenamefont {Gil},\ and\
  \citenamefont {Nazario}}]{bravyi2022future}%
  \BibitemOpen
  \bibfield  {author} {\bibinfo {author} {\bibfnamefont {S.}~\bibnamefont
  {Bravyi}}, \bibinfo {author} {\bibfnamefont {O.}~\bibnamefont {Dial}},
  \bibinfo {author} {\bibfnamefont {J.~M.}\ \bibnamefont {Gambetta}}, \bibinfo
  {author} {\bibfnamefont {D.}~\bibnamefont {Gil}},\ and\ \bibinfo {author}
  {\bibfnamefont {Z.}~\bibnamefont {Nazario}},\ }\bibfield  {title} {\bibinfo
  {title} {The future of quantum computing with superconducting qubits},\
  }\href@noop {} {\bibfield  {journal} {\bibinfo  {journal} {arXiv preprint
  arXiv:2209.06841}\ } (\bibinfo {year} {2022})}\BibitemShut {NoStop}%
\bibitem [{\citenamefont {Bruzewicz}\ \emph {et~al.}(2019)\citenamefont
  {Bruzewicz}, \citenamefont {Chiaverini}, \citenamefont {McConnell},\ and\
  \citenamefont {Sage}}]{bruzewicz2019trapped}%
  \BibitemOpen
  \bibfield  {author} {\bibinfo {author} {\bibfnamefont {C.~D.}\ \bibnamefont
  {Bruzewicz}}, \bibinfo {author} {\bibfnamefont {J.}~\bibnamefont
  {Chiaverini}}, \bibinfo {author} {\bibfnamefont {R.}~\bibnamefont
  {McConnell}},\ and\ \bibinfo {author} {\bibfnamefont {J.~M.}\ \bibnamefont
  {Sage}},\ }\bibfield  {title} {\bibinfo {title} {Trapped-ion quantum
  computing: Progress and challenges},\ }\href@noop {} {\bibfield  {journal}
  {\bibinfo  {journal} {Appl. Phys. Rev.}\ }\textbf {\bibinfo {volume} {6}}
  (\bibinfo {year} {2019})}\BibitemShut {NoStop}%
\bibitem [{\citenamefont {Gross}\ and\ \citenamefont
  {Bloch}(2017)}]{gross2017quantum}%
  \BibitemOpen
  \bibfield  {author} {\bibinfo {author} {\bibfnamefont {C.}~\bibnamefont
  {Gross}}\ and\ \bibinfo {author} {\bibfnamefont {I.}~\bibnamefont {Bloch}},\
  }\bibfield  {title} {\bibinfo {title} {Quantum simulations with ultracold
  atoms in optical lattices},\ }\href@noop {} {\bibfield  {journal} {\bibinfo
  {journal} {Science}\ }\textbf {\bibinfo {volume} {357}},\ \bibinfo {pages}
  {995} (\bibinfo {year} {2017})}\BibitemShut {NoStop}%
\bibitem [{\citenamefont {Preskill}(2018)}]{preskill2018quantum}%
  \BibitemOpen
  \bibfield  {author} {\bibinfo {author} {\bibfnamefont {J.}~\bibnamefont
  {Preskill}},\ }\bibfield  {title} {\bibinfo {title} {Quantum computing in the
  nisq era and beyond},\ }\href@noop {} {\bibfield  {journal} {\bibinfo
  {journal} {Quantum}\ }\textbf {\bibinfo {volume} {2}},\ \bibinfo {pages} {79}
  (\bibinfo {year} {2018})}\BibitemShut {NoStop}%
\bibitem [{\citenamefont {Bharti}\ \emph {et~al.}(2021)\citenamefont {Bharti},
  \citenamefont {Cervera-Lierta}, \citenamefont {Kyaw}, \citenamefont {Haug},
  \citenamefont {Alperin-Lea}, \citenamefont {Anand}, \citenamefont {Degroote},
  \citenamefont {Heimonen}, \citenamefont {Kottmann}, \citenamefont {Menke}
  \emph {et~al.}}]{bharti2021noisy}%
  \BibitemOpen
  \bibfield  {author} {\bibinfo {author} {\bibfnamefont {K.}~\bibnamefont
  {Bharti}}, \bibinfo {author} {\bibfnamefont {A.}~\bibnamefont
  {Cervera-Lierta}}, \bibinfo {author} {\bibfnamefont {T.~H.}\ \bibnamefont
  {Kyaw}}, \bibinfo {author} {\bibfnamefont {T.}~\bibnamefont {Haug}}, \bibinfo
  {author} {\bibfnamefont {S.}~\bibnamefont {Alperin-Lea}}, \bibinfo {author}
  {\bibfnamefont {A.}~\bibnamefont {Anand}}, \bibinfo {author} {\bibfnamefont
  {M.}~\bibnamefont {Degroote}}, \bibinfo {author} {\bibfnamefont
  {H.}~\bibnamefont {Heimonen}}, \bibinfo {author} {\bibfnamefont {J.~S.}\
  \bibnamefont {Kottmann}}, \bibinfo {author} {\bibfnamefont {T.}~\bibnamefont
  {Menke}}, \emph {et~al.},\ }\bibfield  {title} {\bibinfo {title} {Noisy
  intermediate-scale quantum (nisq) algorithms},\ }\href@noop {} {\bibfield
  {journal} {\bibinfo  {journal} {arXiv preprint arXiv:2101.08448}\ } (\bibinfo
  {year} {2021})}\BibitemShut {NoStop}%
\bibitem [{\citenamefont {Kitaev}(2003)}]{kitaev2003fault}%
  \BibitemOpen
  \bibfield  {author} {\bibinfo {author} {\bibfnamefont {A.~Y.}\ \bibnamefont
  {Kitaev}},\ }\bibfield  {title} {\bibinfo {title} {Fault-tolerant quantum
  computation by anyons},\ }\href@noop {} {\bibfield  {journal} {\bibinfo
  {journal} {Ann. Phys. (New York)}\ }\textbf {\bibinfo {volume} {303}},\
  \bibinfo {pages} {2} (\bibinfo {year} {2003})}\BibitemShut {NoStop}%
\bibitem [{\citenamefont {Bravyi}\ and\ \citenamefont
  {Kitaev}(1998)}]{bravyi1998quantum}%
  \BibitemOpen
  \bibfield  {author} {\bibinfo {author} {\bibfnamefont {S.~B.}\ \bibnamefont
  {Bravyi}}\ and\ \bibinfo {author} {\bibfnamefont {A.~Y.}\ \bibnamefont
  {Kitaev}},\ }\bibfield  {title} {\bibinfo {title} {Quantum codes on a lattice
  with boundary},\ }\href@noop {} {\bibfield  {journal} {\bibinfo  {journal}
  {arXiv preprint quant-ph/9811052}\ } (\bibinfo {year} {1998})}\BibitemShut
  {NoStop}%
\bibitem [{\citenamefont {Fowler}\ \emph {et~al.}(2012)\citenamefont {Fowler},
  \citenamefont {Mariantoni}, \citenamefont {Martinis},\ and\ \citenamefont
  {Cleland}}]{fowler2012surface}%
  \BibitemOpen
  \bibfield  {author} {\bibinfo {author} {\bibfnamefont {A.~G.}\ \bibnamefont
  {Fowler}}, \bibinfo {author} {\bibfnamefont {M.}~\bibnamefont {Mariantoni}},
  \bibinfo {author} {\bibfnamefont {J.~M.}\ \bibnamefont {Martinis}},\ and\
  \bibinfo {author} {\bibfnamefont {A.~N.}\ \bibnamefont {Cleland}},\
  }\bibfield  {title} {\bibinfo {title} {Surface codes: Towards practical
  large-scale quantum computation},\ }\href@noop {} {\bibfield  {journal}
  {\bibinfo  {journal} {Phys. Rev. A}\ }\textbf {\bibinfo {volume} {86}},\
  \bibinfo {pages} {032324} (\bibinfo {year} {2012})}\BibitemShut {NoStop}%
\bibitem [{\citenamefont {Dennis}\ \emph {et~al.}(2002)\citenamefont {Dennis},
  \citenamefont {Kitaev}, \citenamefont {Landahl},\ and\ \citenamefont
  {Preskill}}]{dennis2002topological}%
  \BibitemOpen
  \bibfield  {author} {\bibinfo {author} {\bibfnamefont {E.}~\bibnamefont
  {Dennis}}, \bibinfo {author} {\bibfnamefont {A.}~\bibnamefont {Kitaev}},
  \bibinfo {author} {\bibfnamefont {A.}~\bibnamefont {Landahl}},\ and\ \bibinfo
  {author} {\bibfnamefont {J.}~\bibnamefont {Preskill}},\ }\bibfield  {title}
  {\bibinfo {title} {Topological quantum memory},\ }\href@noop {} {\bibfield
  {journal} {\bibinfo  {journal} {J. Math. Phys.}\ }\textbf {\bibinfo {volume}
  {43}},\ \bibinfo {pages} {4452} (\bibinfo {year} {2002})}\BibitemShut
  {NoStop}%
\bibitem [{\citenamefont {Raussendorf}\ and\ \citenamefont
  {Harrington}(2007)}]{raussendorf2007fault}%
  \BibitemOpen
  \bibfield  {author} {\bibinfo {author} {\bibfnamefont {R.}~\bibnamefont
  {Raussendorf}}\ and\ \bibinfo {author} {\bibfnamefont {J.}~\bibnamefont
  {Harrington}},\ }\bibfield  {title} {\bibinfo {title} {Fault-tolerant quantum
  computation with high threshold in two dimensions},\ }\href@noop {}
  {\bibfield  {journal} {\bibinfo  {journal} {Phys. Rev. Lett.}\ }\textbf
  {\bibinfo {volume} {98}},\ \bibinfo {pages} {190504} (\bibinfo {year}
  {2007})}\BibitemShut {NoStop}%
\bibitem [{\citenamefont {Satzinger}\ \emph {et~al.}(2021)\citenamefont
  {Satzinger}, \citenamefont {Liu}, \citenamefont {Smith}, \citenamefont
  {Knapp}, \citenamefont {Newman}, \citenamefont {Jones}, \citenamefont {Chen},
  \citenamefont {Quintana}, \citenamefont {Mi}, \citenamefont {Dunsworth} \emph
  {et~al.}}]{satzinger2021realizing}%
  \BibitemOpen
  \bibfield  {author} {\bibinfo {author} {\bibfnamefont {K.}~\bibnamefont
  {Satzinger}}, \bibinfo {author} {\bibfnamefont {Y.-J.}\ \bibnamefont {Liu}},
  \bibinfo {author} {\bibfnamefont {A.}~\bibnamefont {Smith}}, \bibinfo
  {author} {\bibfnamefont {C.}~\bibnamefont {Knapp}}, \bibinfo {author}
  {\bibfnamefont {M.}~\bibnamefont {Newman}}, \bibinfo {author} {\bibfnamefont
  {C.}~\bibnamefont {Jones}}, \bibinfo {author} {\bibfnamefont
  {Z.}~\bibnamefont {Chen}}, \bibinfo {author} {\bibfnamefont {C.}~\bibnamefont
  {Quintana}}, \bibinfo {author} {\bibfnamefont {X.}~\bibnamefont {Mi}},
  \bibinfo {author} {\bibfnamefont {A.}~\bibnamefont {Dunsworth}}, \emph
  {et~al.},\ }\bibfield  {title} {\bibinfo {title} {Realizing topologically
  ordered states on a quantum processor},\ }\href@noop {} {\bibfield  {journal}
  {\bibinfo  {journal} {Science}\ }\textbf {\bibinfo {volume} {374}},\ \bibinfo
  {pages} {1237} (\bibinfo {year} {2021})}\BibitemShut {NoStop}%
\bibitem [{\citenamefont {Semeghini}\ \emph {et~al.}(2021)\citenamefont
  {Semeghini}, \citenamefont {Levine}, \citenamefont {Keesling}, \citenamefont
  {Ebadi}, \citenamefont {Wang}, \citenamefont {Bluvstein}, \citenamefont
  {Verresen}, \citenamefont {Pichler}, \citenamefont {Kalinowski},
  \citenamefont {Samajdar} \emph {et~al.}}]{semeghini2021probing}%
  \BibitemOpen
  \bibfield  {author} {\bibinfo {author} {\bibfnamefont {G.}~\bibnamefont
  {Semeghini}}, \bibinfo {author} {\bibfnamefont {H.}~\bibnamefont {Levine}},
  \bibinfo {author} {\bibfnamefont {A.}~\bibnamefont {Keesling}}, \bibinfo
  {author} {\bibfnamefont {S.}~\bibnamefont {Ebadi}}, \bibinfo {author}
  {\bibfnamefont {T.~T.}\ \bibnamefont {Wang}}, \bibinfo {author}
  {\bibfnamefont {D.}~\bibnamefont {Bluvstein}}, \bibinfo {author}
  {\bibfnamefont {R.}~\bibnamefont {Verresen}}, \bibinfo {author}
  {\bibfnamefont {H.}~\bibnamefont {Pichler}}, \bibinfo {author} {\bibfnamefont
  {M.}~\bibnamefont {Kalinowski}}, \bibinfo {author} {\bibfnamefont
  {R.}~\bibnamefont {Samajdar}}, \emph {et~al.},\ }\bibfield  {title} {\bibinfo
  {title} {Probing topological spin liquids on a programmable quantum
  simulator},\ }\href@noop {} {\bibfield  {journal} {\bibinfo  {journal}
  {Science}\ }\textbf {\bibinfo {volume} {374}},\ \bibinfo {pages} {1242}
  (\bibinfo {year} {2021})}\BibitemShut {NoStop}%
\bibitem [{\citenamefont {Shor}(1995)}]{shor1995scheme}%
  \BibitemOpen
  \bibfield  {author} {\bibinfo {author} {\bibfnamefont {P.~W.}\ \bibnamefont
  {Shor}},\ }\bibfield  {title} {\bibinfo {title} {Scheme for reducing
  decoherence in quantum computer memory},\ }\href@noop {} {\bibfield
  {journal} {\bibinfo  {journal} {Phys. Rev. A}\ }\textbf {\bibinfo {volume}
  {52}},\ \bibinfo {pages} {R2493} (\bibinfo {year} {1995})}\BibitemShut
  {NoStop}%
\bibitem [{\citenamefont {Calderbank}\ and\ \citenamefont
  {Shor}(1996)}]{calderbank1996good}%
  \BibitemOpen
  \bibfield  {author} {\bibinfo {author} {\bibfnamefont {A.~R.}\ \bibnamefont
  {Calderbank}}\ and\ \bibinfo {author} {\bibfnamefont {P.~W.}\ \bibnamefont
  {Shor}},\ }\bibfield  {title} {\bibinfo {title} {Good quantum
  error-correcting codes exist},\ }\href@noop {} {\bibfield  {journal}
  {\bibinfo  {journal} {Phys. Rev. A}\ }\textbf {\bibinfo {volume} {54}},\
  \bibinfo {pages} {1098} (\bibinfo {year} {1996})}\BibitemShut {NoStop}%
\bibitem [{\citenamefont {Steane}(1996)}]{steane1996error}%
  \BibitemOpen
  \bibfield  {author} {\bibinfo {author} {\bibfnamefont {A.~M.}\ \bibnamefont
  {Steane}},\ }\bibfield  {title} {\bibinfo {title} {Error correcting codes in
  quantum theory},\ }\href@noop {} {\bibfield  {journal} {\bibinfo  {journal}
  {Phys. Rev. Lett.}\ }\textbf {\bibinfo {volume} {77}},\ \bibinfo {pages}
  {793} (\bibinfo {year} {1996})}\BibitemShut {NoStop}%
\bibitem [{\citenamefont {Calderbank}\ \emph {et~al.}(1998)\citenamefont
  {Calderbank}, \citenamefont {Rains}, \citenamefont {Shor},\ and\
  \citenamefont {Sloane}}]{calderbank1998quantum}%
  \BibitemOpen
  \bibfield  {author} {\bibinfo {author} {\bibfnamefont {A.~R.}\ \bibnamefont
  {Calderbank}}, \bibinfo {author} {\bibfnamefont {E.~M.}\ \bibnamefont
  {Rains}}, \bibinfo {author} {\bibfnamefont {P.}~\bibnamefont {Shor}},\ and\
  \bibinfo {author} {\bibfnamefont {N.~J.}\ \bibnamefont {Sloane}},\ }\bibfield
   {title} {\bibinfo {title} {Quantum error correction via codes over gf (4)},\
  }\href@noop {} {\bibfield  {journal} {\bibinfo  {journal} {IEEE T. Inform.
  Theory}\ }\textbf {\bibinfo {volume} {44}},\ \bibinfo {pages} {1369}
  (\bibinfo {year} {1998})}\BibitemShut {NoStop}%
\bibitem [{\citenamefont {Gottesman}(1997)}]{gottesman1997stabilizer}%
  \BibitemOpen
  \bibfield  {author} {\bibinfo {author} {\bibfnamefont {D.}~\bibnamefont
  {Gottesman}},\ }\href@noop {} {\emph {\bibinfo {title} {Stabilizer codes and
  quantum error correction}}}\ (\bibinfo  {publisher} {California Institute of
  Technology},\ \bibinfo {year} {1997})\BibitemShut {NoStop}%
\bibitem [{\citenamefont {Vincent}\ \emph {et~al.}(2010)\citenamefont
  {Vincent}, \citenamefont {Larochelle}, \citenamefont {Lajoie}, \citenamefont
  {Bengio}, \citenamefont {Manzagol},\ and\ \citenamefont
  {Bottou}}]{vincent2010stacked}%
  \BibitemOpen
  \bibfield  {author} {\bibinfo {author} {\bibfnamefont {P.}~\bibnamefont
  {Vincent}}, \bibinfo {author} {\bibfnamefont {H.}~\bibnamefont {Larochelle}},
  \bibinfo {author} {\bibfnamefont {I.}~\bibnamefont {Lajoie}}, \bibinfo
  {author} {\bibfnamefont {Y.}~\bibnamefont {Bengio}}, \bibinfo {author}
  {\bibfnamefont {P.-A.}\ \bibnamefont {Manzagol}},\ and\ \bibinfo {author}
  {\bibfnamefont {L.}~\bibnamefont {Bottou}},\ }\bibfield  {title} {\bibinfo
  {title} {Stacked denoising autoencoders: Learning useful representations in a
  deep network with a local denoising criterion.},\ }\href@noop {} {\bibfield
  {journal} {\bibinfo  {journal} {J. Mach. Learn. Res.}\ }\textbf {\bibinfo
  {volume} {11}} (\bibinfo {year} {2010})}\BibitemShut {NoStop}%
\bibitem [{\citenamefont {Krizhevsky}\ \emph {et~al.}(2017)\citenamefont
  {Krizhevsky}, \citenamefont {Sutskever},\ and\ \citenamefont
  {Hinton}}]{krizhevsky2017imagenet}%
  \BibitemOpen
  \bibfield  {author} {\bibinfo {author} {\bibfnamefont {A.}~\bibnamefont
  {Krizhevsky}}, \bibinfo {author} {\bibfnamefont {I.}~\bibnamefont
  {Sutskever}},\ and\ \bibinfo {author} {\bibfnamefont {G.~E.}\ \bibnamefont
  {Hinton}},\ }\bibfield  {title} {\bibinfo {title} {Imagenet classification
  with deep convolutional neural networks},\ }\href@noop {} {\bibfield
  {journal} {\bibinfo  {journal} {Commun. ACM}\ }\textbf {\bibinfo {volume}
  {60}},\ \bibinfo {pages} {84} (\bibinfo {year} {2017})}\BibitemShut {NoStop}%
\bibitem [{\citenamefont {Kingma}\ and\ \citenamefont
  {Welling}(2013)}]{kingma2013auto}%
  \BibitemOpen
  \bibfield  {author} {\bibinfo {author} {\bibfnamefont {D.~P.}\ \bibnamefont
  {Kingma}}\ and\ \bibinfo {author} {\bibfnamefont {M.}~\bibnamefont
  {Welling}},\ }\bibfield  {title} {\bibinfo {title} {Auto-encoding variational
  bayes},\ }\href@noop {} {\bibfield  {journal} {\bibinfo  {journal} {arXiv
  preprint arXiv:1312.6114}\ } (\bibinfo {year} {2013})}\BibitemShut {NoStop}%
\bibitem [{\citenamefont {Goodfellow}\ \emph {et~al.}(2016)\citenamefont
  {Goodfellow}, \citenamefont {Bengio},\ and\ \citenamefont
  {Courville}}]{goodfellow2016deep}%
  \BibitemOpen
  \bibfield  {author} {\bibinfo {author} {\bibfnamefont {I.}~\bibnamefont
  {Goodfellow}}, \bibinfo {author} {\bibfnamefont {Y.}~\bibnamefont {Bengio}},\
  and\ \bibinfo {author} {\bibfnamefont {A.}~\bibnamefont {Courville}},\
  }\href@noop {} {\emph {\bibinfo {title} {Deep learning}}}\ (\bibinfo
  {publisher} {MIT press},\ \bibinfo {year} {2016})\BibitemShut {NoStop}%
\bibitem [{\citenamefont {Goodfellow}\ \emph {et~al.}(2020)\citenamefont
  {Goodfellow}, \citenamefont {Pouget-Abadie}, \citenamefont {Mirza},
  \citenamefont {Xu}, \citenamefont {Warde-Farley}, \citenamefont {Ozair},
  \citenamefont {Courville},\ and\ \citenamefont
  {Bengio}}]{goodfellow2020generative}%
  \BibitemOpen
  \bibfield  {author} {\bibinfo {author} {\bibfnamefont {I.}~\bibnamefont
  {Goodfellow}}, \bibinfo {author} {\bibfnamefont {J.}~\bibnamefont
  {Pouget-Abadie}}, \bibinfo {author} {\bibfnamefont {M.}~\bibnamefont
  {Mirza}}, \bibinfo {author} {\bibfnamefont {B.}~\bibnamefont {Xu}}, \bibinfo
  {author} {\bibfnamefont {D.}~\bibnamefont {Warde-Farley}}, \bibinfo {author}
  {\bibfnamefont {S.}~\bibnamefont {Ozair}}, \bibinfo {author} {\bibfnamefont
  {A.}~\bibnamefont {Courville}},\ and\ \bibinfo {author} {\bibfnamefont
  {Y.}~\bibnamefont {Bengio}},\ }\bibfield  {title} {\bibinfo {title}
  {Generative adversarial networks},\ }\href@noop {} {\bibfield  {journal}
  {\bibinfo  {journal} {Commun. ACM}\ }\textbf {\bibinfo {volume} {63}},\
  \bibinfo {pages} {139} (\bibinfo {year} {2020})}\BibitemShut {NoStop}%
\bibitem [{\citenamefont {Carleo}\ and\ \citenamefont
  {Troyer}(2017)}]{carleo2017solving}%
  \BibitemOpen
  \bibfield  {author} {\bibinfo {author} {\bibfnamefont {G.}~\bibnamefont
  {Carleo}}\ and\ \bibinfo {author} {\bibfnamefont {M.}~\bibnamefont
  {Troyer}},\ }\bibfield  {title} {\bibinfo {title} {Solving the quantum
  many-body problem with artificial neural networks},\ }\href@noop {}
  {\bibfield  {journal} {\bibinfo  {journal} {Science}\ }\textbf {\bibinfo
  {volume} {355}},\ \bibinfo {pages} {602} (\bibinfo {year}
  {2017})}\BibitemShut {NoStop}%
\bibitem [{\citenamefont {Biamonte}\ \emph {et~al.}(2017)\citenamefont
  {Biamonte}, \citenamefont {Wittek}, \citenamefont {Pancotti}, \citenamefont
  {Rebentrost}, \citenamefont {Wiebe},\ and\ \citenamefont
  {Lloyd}}]{biamonte2017quantum}%
  \BibitemOpen
  \bibfield  {author} {\bibinfo {author} {\bibfnamefont {J.}~\bibnamefont
  {Biamonte}}, \bibinfo {author} {\bibfnamefont {P.}~\bibnamefont {Wittek}},
  \bibinfo {author} {\bibfnamefont {N.}~\bibnamefont {Pancotti}}, \bibinfo
  {author} {\bibfnamefont {P.}~\bibnamefont {Rebentrost}}, \bibinfo {author}
  {\bibfnamefont {N.}~\bibnamefont {Wiebe}},\ and\ \bibinfo {author}
  {\bibfnamefont {S.}~\bibnamefont {Lloyd}},\ }\bibfield  {title} {\bibinfo
  {title} {Quantum machine learning},\ }\href@noop {} {\bibfield  {journal}
  {\bibinfo  {journal} {Nature}\ }\textbf {\bibinfo {volume} {549}},\ \bibinfo
  {pages} {195} (\bibinfo {year} {2017})}\BibitemShut {NoStop}%
\bibitem [{\citenamefont {Van~Nieuwenburg}\ \emph {et~al.}(2017)\citenamefont
  {Van~Nieuwenburg}, \citenamefont {Liu},\ and\ \citenamefont
  {Huber}}]{van2017learning}%
  \BibitemOpen
  \bibfield  {author} {\bibinfo {author} {\bibfnamefont {E.~P.}\ \bibnamefont
  {Van~Nieuwenburg}}, \bibinfo {author} {\bibfnamefont {Y.-H.}\ \bibnamefont
  {Liu}},\ and\ \bibinfo {author} {\bibfnamefont {S.~D.}\ \bibnamefont
  {Huber}},\ }\bibfield  {title} {\bibinfo {title} {Learning phase transitions
  by confusion},\ }\href@noop {} {\bibfield  {journal} {\bibinfo  {journal}
  {Nat. Phys.}\ }\textbf {\bibinfo {volume} {13}},\ \bibinfo {pages} {435}
  (\bibinfo {year} {2017})}\BibitemShut {NoStop}%
\bibitem [{\citenamefont {Carrasquilla}\ and\ \citenamefont
  {Melko}(2017)}]{carrasquilla2017machine}%
  \BibitemOpen
  \bibfield  {author} {\bibinfo {author} {\bibfnamefont {J.}~\bibnamefont
  {Carrasquilla}}\ and\ \bibinfo {author} {\bibfnamefont {R.~G.}\ \bibnamefont
  {Melko}},\ }\bibfield  {title} {\bibinfo {title} {Machine learning phases of
  matter},\ }\href@noop {} {\bibfield  {journal} {\bibinfo  {journal} {Nat.
  Phys.}\ }\textbf {\bibinfo {volume} {13}},\ \bibinfo {pages} {431} (\bibinfo
  {year} {2017})}\BibitemShut {NoStop}%
\bibitem [{\citenamefont {Carleo}\ \emph {et~al.}(2019)\citenamefont {Carleo},
  \citenamefont {Cirac}, \citenamefont {Cranmer}, \citenamefont {Daudet},
  \citenamefont {Schuld}, \citenamefont {Tishby}, \citenamefont
  {Vogt-Maranto},\ and\ \citenamefont {Zdeborov{\'a}}}]{carleo2019machine}%
  \BibitemOpen
  \bibfield  {author} {\bibinfo {author} {\bibfnamefont {G.}~\bibnamefont
  {Carleo}}, \bibinfo {author} {\bibfnamefont {I.}~\bibnamefont {Cirac}},
  \bibinfo {author} {\bibfnamefont {K.}~\bibnamefont {Cranmer}}, \bibinfo
  {author} {\bibfnamefont {L.}~\bibnamefont {Daudet}}, \bibinfo {author}
  {\bibfnamefont {M.}~\bibnamefont {Schuld}}, \bibinfo {author} {\bibfnamefont
  {N.}~\bibnamefont {Tishby}}, \bibinfo {author} {\bibfnamefont
  {L.}~\bibnamefont {Vogt-Maranto}},\ and\ \bibinfo {author} {\bibfnamefont
  {L.}~\bibnamefont {Zdeborov{\'a}}},\ }\bibfield  {title} {\bibinfo {title}
  {Machine learning and the physical sciences},\ }\href@noop {} {\bibfield
  {journal} {\bibinfo  {journal} {Rev. Mod. Phys.}\ }\textbf {\bibinfo {volume}
  {91}},\ \bibinfo {pages} {045002} (\bibinfo {year} {2019})}\BibitemShut
  {NoStop}%
\bibitem [{\citenamefont {Wang}(2016)}]{wang2016discovering}%
  \BibitemOpen
  \bibfield  {author} {\bibinfo {author} {\bibfnamefont {L.}~\bibnamefont
  {Wang}},\ }\bibfield  {title} {\bibinfo {title} {Discovering phase
  transitions with unsupervised learning},\ }\href@noop {} {\bibfield
  {journal} {\bibinfo  {journal} {Phys. Rev. B}\ }\textbf {\bibinfo {volume}
  {94}},\ \bibinfo {pages} {195105} (\bibinfo {year} {2016})}\BibitemShut
  {NoStop}%
\bibitem [{\citenamefont {Zhang}\ and\ \citenamefont
  {Kim}(2017)}]{zhang2017quantum}%
  \BibitemOpen
  \bibfield  {author} {\bibinfo {author} {\bibfnamefont {Y.}~\bibnamefont
  {Zhang}}\ and\ \bibinfo {author} {\bibfnamefont {E.-A.}\ \bibnamefont
  {Kim}},\ }\bibfield  {title} {\bibinfo {title} {Quantum loop topography for
  machine learning},\ }\href@noop {} {\bibfield  {journal} {\bibinfo  {journal}
  {Phys. Rev. Lett.}\ }\textbf {\bibinfo {volume} {118}},\ \bibinfo {pages}
  {216401} (\bibinfo {year} {2017})}\BibitemShut {NoStop}%
\bibitem [{\citenamefont {Cong}\ \emph {et~al.}(2019)\citenamefont {Cong},
  \citenamefont {Choi},\ and\ \citenamefont {Lukin}}]{cong2019quantum}%
  \BibitemOpen
  \bibfield  {author} {\bibinfo {author} {\bibfnamefont {I.}~\bibnamefont
  {Cong}}, \bibinfo {author} {\bibfnamefont {S.}~\bibnamefont {Choi}},\ and\
  \bibinfo {author} {\bibfnamefont {M.~D.}\ \bibnamefont {Lukin}},\ }\bibfield
  {title} {\bibinfo {title} {Quantum convolutional neural networks},\
  }\href@noop {} {\bibfield  {journal} {\bibinfo  {journal} {Nat. Phys.}\
  }\textbf {\bibinfo {volume} {15}},\ \bibinfo {pages} {1273} (\bibinfo {year}
  {2019})}\BibitemShut {NoStop}%
\bibitem [{\citenamefont {Han}\ \emph {et~al.}(2018)\citenamefont {Han},
  \citenamefont {Wang}, \citenamefont {Fan}, \citenamefont {Wang},\ and\
  \citenamefont {Zhang}}]{han2018unsupervised}%
  \BibitemOpen
  \bibfield  {author} {\bibinfo {author} {\bibfnamefont {Z.-Y.}\ \bibnamefont
  {Han}}, \bibinfo {author} {\bibfnamefont {J.}~\bibnamefont {Wang}}, \bibinfo
  {author} {\bibfnamefont {H.}~\bibnamefont {Fan}}, \bibinfo {author}
  {\bibfnamefont {L.}~\bibnamefont {Wang}},\ and\ \bibinfo {author}
  {\bibfnamefont {P.}~\bibnamefont {Zhang}},\ }\bibfield  {title} {\bibinfo
  {title} {Unsupervised generative modeling using matrix product states},\
  }\href@noop {} {\bibfield  {journal} {\bibinfo  {journal} {Phys. Rev. X}\
  }\textbf {\bibinfo {volume} {8}},\ \bibinfo {pages} {031012} (\bibinfo {year}
  {2018})}\BibitemShut {NoStop}%
\bibitem [{\citenamefont {Gomez}\ \emph {et~al.}(2022)\citenamefont {Gomez},
  \citenamefont {Yelin},\ and\ \citenamefont
  {Najafi}}]{gomez2022reconstructing}%
  \BibitemOpen
  \bibfield  {author} {\bibinfo {author} {\bibfnamefont {A.~M.}\ \bibnamefont
  {Gomez}}, \bibinfo {author} {\bibfnamefont {S.~F.}\ \bibnamefont {Yelin}},\
  and\ \bibinfo {author} {\bibfnamefont {K.}~\bibnamefont {Najafi}},\
  }\bibfield  {title} {\bibinfo {title} {Reconstructing quantum states using
  basis-enhanced born machines},\ }\href@noop {} {\bibfield  {journal}
  {\bibinfo  {journal} {arXiv preprint arXiv:2206.01273}\ } (\bibinfo {year}
  {2022})}\BibitemShut {NoStop}%
\bibitem [{\citenamefont {Zhong}\ \emph {et~al.}(2022)\citenamefont {Zhong},
  \citenamefont {Gao}, \citenamefont {Yelin},\ and\ \citenamefont
  {Najafi}}]{zhong2022many}%
  \BibitemOpen
  \bibfield  {author} {\bibinfo {author} {\bibfnamefont {W.}~\bibnamefont
  {Zhong}}, \bibinfo {author} {\bibfnamefont {X.}~\bibnamefont {Gao}}, \bibinfo
  {author} {\bibfnamefont {S.~F.}\ \bibnamefont {Yelin}},\ and\ \bibinfo
  {author} {\bibfnamefont {K.}~\bibnamefont {Najafi}},\ }\bibfield  {title}
  {\bibinfo {title} {Many-body localized hidden born machine},\ }\href@noop {}
  {\bibfield  {journal} {\bibinfo  {journal} {arXiv preprint arXiv:2207.02346}\
  } (\bibinfo {year} {2022})}\BibitemShut {NoStop}%
\bibitem [{\citenamefont {Benedetti}\ \emph {et~al.}(2019)\citenamefont
  {Benedetti}, \citenamefont {Lloyd}, \citenamefont {Sack},\ and\ \citenamefont
  {Fiorentini}}]{benedetti2019parameterized}%
  \BibitemOpen
  \bibfield  {author} {\bibinfo {author} {\bibfnamefont {M.}~\bibnamefont
  {Benedetti}}, \bibinfo {author} {\bibfnamefont {E.}~\bibnamefont {Lloyd}},
  \bibinfo {author} {\bibfnamefont {S.}~\bibnamefont {Sack}},\ and\ \bibinfo
  {author} {\bibfnamefont {M.}~\bibnamefont {Fiorentini}},\ }\bibfield  {title}
  {\bibinfo {title} {Parameterized quantum circuits as machine learning
  models},\ }\href@noop {} {\bibfield  {journal} {\bibinfo  {journal} {Quantum
  Sci. Technol.}\ }\textbf {\bibinfo {volume} {4}},\ \bibinfo {pages} {043001}
  (\bibinfo {year} {2019})}\BibitemShut {NoStop}%
\bibitem [{\citenamefont {Abbas}\ \emph {et~al.}(2021)\citenamefont {Abbas},
  \citenamefont {Sutter}, \citenamefont {Zoufal}, \citenamefont {Lucchi},
  \citenamefont {Figalli},\ and\ \citenamefont {Woerner}}]{abbas2021power}%
  \BibitemOpen
  \bibfield  {author} {\bibinfo {author} {\bibfnamefont {A.}~\bibnamefont
  {Abbas}}, \bibinfo {author} {\bibfnamefont {D.}~\bibnamefont {Sutter}},
  \bibinfo {author} {\bibfnamefont {C.}~\bibnamefont {Zoufal}}, \bibinfo
  {author} {\bibfnamefont {A.}~\bibnamefont {Lucchi}}, \bibinfo {author}
  {\bibfnamefont {A.}~\bibnamefont {Figalli}},\ and\ \bibinfo {author}
  {\bibfnamefont {S.}~\bibnamefont {Woerner}},\ }\bibfield  {title} {\bibinfo
  {title} {The power of quantum neural networks},\ }\href@noop {} {\bibfield
  {journal} {\bibinfo  {journal} {Nat. Comput. Sci.}\ }\textbf {\bibinfo
  {volume} {1}},\ \bibinfo {pages} {403} (\bibinfo {year} {2021})}\BibitemShut
  {NoStop}%
\bibitem [{\citenamefont {Locher}\ \emph {et~al.}(2023)\citenamefont {Locher},
  \citenamefont {Cardarelli},\ and\ \citenamefont
  {M{\"u}ller}}]{locher2023quantum}%
  \BibitemOpen
  \bibfield  {author} {\bibinfo {author} {\bibfnamefont {D.~F.}\ \bibnamefont
  {Locher}}, \bibinfo {author} {\bibfnamefont {L.}~\bibnamefont {Cardarelli}},\
  and\ \bibinfo {author} {\bibfnamefont {M.}~\bibnamefont {M{\"u}ller}},\
  }\bibfield  {title} {\bibinfo {title} {Quantum error correction with quantum
  autoencoders},\ }\href@noop {} {\bibfield  {journal} {\bibinfo  {journal}
  {Quantum}\ }\textbf {\bibinfo {volume} {7}},\ \bibinfo {pages} {942}
  (\bibinfo {year} {2023})}\BibitemShut {NoStop}%
\bibitem [{\citenamefont {Cao}\ \emph {et~al.}(2022)\citenamefont {Cao},
  \citenamefont {Zhang}, \citenamefont {Wu}, \citenamefont {Grassl},\ and\
  \citenamefont {Zeng}}]{cao2022quantum}%
  \BibitemOpen
  \bibfield  {author} {\bibinfo {author} {\bibfnamefont {C.}~\bibnamefont
  {Cao}}, \bibinfo {author} {\bibfnamefont {C.}~\bibnamefont {Zhang}}, \bibinfo
  {author} {\bibfnamefont {Z.}~\bibnamefont {Wu}}, \bibinfo {author}
  {\bibfnamefont {M.}~\bibnamefont {Grassl}},\ and\ \bibinfo {author}
  {\bibfnamefont {B.}~\bibnamefont {Zeng}},\ }\bibfield  {title} {\bibinfo
  {title} {Quantum variational learning for quantum error-correcting codes},\
  }\href@noop {} {\bibfield  {journal} {\bibinfo  {journal} {Quantum}\ }\textbf
  {\bibinfo {volume} {6}},\ \bibinfo {pages} {828} (\bibinfo {year}
  {2022})}\BibitemShut {NoStop}%
\bibitem [{\citenamefont {McClean}\ \emph {et~al.}(2018)\citenamefont
  {McClean}, \citenamefont {Boixo}, \citenamefont {Smelyanskiy}, \citenamefont
  {Babbush},\ and\ \citenamefont {Neven}}]{mcclean2018barren}%
  \BibitemOpen
  \bibfield  {author} {\bibinfo {author} {\bibfnamefont {J.~R.}\ \bibnamefont
  {McClean}}, \bibinfo {author} {\bibfnamefont {S.}~\bibnamefont {Boixo}},
  \bibinfo {author} {\bibfnamefont {V.~N.}\ \bibnamefont {Smelyanskiy}},
  \bibinfo {author} {\bibfnamefont {R.}~\bibnamefont {Babbush}},\ and\ \bibinfo
  {author} {\bibfnamefont {H.}~\bibnamefont {Neven}},\ }\bibfield  {title}
  {\bibinfo {title} {Barren plateaus in quantum neural network training
  landscapes},\ }\href@noop {} {\bibfield  {journal} {\bibinfo  {journal} {Nat.
  Commun.}\ }\textbf {\bibinfo {volume} {9}},\ \bibinfo {pages} {1} (\bibinfo
  {year} {2018})}\BibitemShut {NoStop}%
\bibitem [{\citenamefont {Cerezo}\ \emph {et~al.}(2021)\citenamefont {Cerezo},
  \citenamefont {Sone}, \citenamefont {Volkoff}, \citenamefont {Cincio},\ and\
  \citenamefont {Coles}}]{cerezo2021cost}%
  \BibitemOpen
  \bibfield  {author} {\bibinfo {author} {\bibfnamefont {M.}~\bibnamefont
  {Cerezo}}, \bibinfo {author} {\bibfnamefont {A.}~\bibnamefont {Sone}},
  \bibinfo {author} {\bibfnamefont {T.}~\bibnamefont {Volkoff}}, \bibinfo
  {author} {\bibfnamefont {L.}~\bibnamefont {Cincio}},\ and\ \bibinfo {author}
  {\bibfnamefont {P.~J.}\ \bibnamefont {Coles}},\ }\bibfield  {title} {\bibinfo
  {title} {Cost function dependent barren plateaus in shallow parametrized
  quantum circuits},\ }\href@noop {} {\bibfield  {journal} {\bibinfo  {journal}
  {Nat. Commun.}\ }\textbf {\bibinfo {volume} {12}},\ \bibinfo {pages} {1}
  (\bibinfo {year} {2021})}\BibitemShut {NoStop}%
\bibitem [{\citenamefont {Movassagh}\ and\ \citenamefont
  {Ouyang}(2020)}]{movassagh2020constructing}%
  \BibitemOpen
  \bibfield  {author} {\bibinfo {author} {\bibfnamefont {R.}~\bibnamefont
  {Movassagh}}\ and\ \bibinfo {author} {\bibfnamefont {Y.}~\bibnamefont
  {Ouyang}},\ }\bibfield  {title} {\bibinfo {title} {Constructing quantum codes
  from any classical code and their embedding in ground space of local
  hamiltonians},\ }\href@noop {} {\bibfield  {journal} {\bibinfo  {journal}
  {arXiv preprint arXiv:2012.01453}\ } (\bibinfo {year} {2020})}\BibitemShut
  {NoStop}%
\bibitem [{\citenamefont {Bravyi}\ \emph {et~al.}(2010)\citenamefont {Bravyi},
  \citenamefont {Hastings},\ and\ \citenamefont
  {Michalakis}}]{bravyi2010topological}%
  \BibitemOpen
  \bibfield  {author} {\bibinfo {author} {\bibfnamefont {S.}~\bibnamefont
  {Bravyi}}, \bibinfo {author} {\bibfnamefont {M.~B.}\ \bibnamefont
  {Hastings}},\ and\ \bibinfo {author} {\bibfnamefont {S.}~\bibnamefont
  {Michalakis}},\ }\bibfield  {title} {\bibinfo {title} {Topological quantum
  order: stability under local perturbations},\ }\href@noop {} {\bibfield
  {journal} {\bibinfo  {journal} {J. Math. Phys.}\ }\textbf {\bibinfo {volume}
  {51}},\ \bibinfo {pages} {093512} (\bibinfo {year} {2010})}\BibitemShut
  {NoStop}%
\bibitem [{\citenamefont {Laflamme}\ \emph {et~al.}(1996)\citenamefont
  {Laflamme}, \citenamefont {Miquel}, \citenamefont {Paz},\ and\ \citenamefont
  {Zurek}}]{laflamme1996perfect}%
  \BibitemOpen
  \bibfield  {author} {\bibinfo {author} {\bibfnamefont {R.}~\bibnamefont
  {Laflamme}}, \bibinfo {author} {\bibfnamefont {C.}~\bibnamefont {Miquel}},
  \bibinfo {author} {\bibfnamefont {J.~P.}\ \bibnamefont {Paz}},\ and\ \bibinfo
  {author} {\bibfnamefont {W.~H.}\ \bibnamefont {Zurek}},\ }\bibfield  {title}
  {\bibinfo {title} {Perfect quantum error correcting code},\ }\href@noop {}
  {\bibfield  {journal} {\bibinfo  {journal} {Phys. Rev. Lett.}\ }\textbf
  {\bibinfo {volume} {77}},\ \bibinfo {pages} {198} (\bibinfo {year}
  {1996})}\BibitemShut {NoStop}%
\bibitem [{\citenamefont {Grassl}(2007)}]{Grassl:codetables}%
  \BibitemOpen
  \bibfield  {author} {\bibinfo {author} {\bibfnamefont {M.}~\bibnamefont
  {Grassl}},\ }\href@noop {} {\bibinfo {title} {{Bounds on the minimum distance
  of linear codes and quantum codes}}},\ \bibinfo {howpublished} {Online
  available at \url{http://www.codetables.de}} (\bibinfo {year} {2007}),\
  \bibinfo {note} {accessed on 2023-04-11}\BibitemShut {NoStop}%
\bibitem [{\citenamefont {Huang}\ \emph {et~al.}(2021)\citenamefont {Huang},
  \citenamefont {Kueng}, \citenamefont {Torlai}, \citenamefont {Albert},\ and\
  \citenamefont {Preskill}}]{huang2021provably}%
  \BibitemOpen
  \bibfield  {author} {\bibinfo {author} {\bibfnamefont {H.-Y.}\ \bibnamefont
  {Huang}}, \bibinfo {author} {\bibfnamefont {R.}~\bibnamefont {Kueng}},
  \bibinfo {author} {\bibfnamefont {G.}~\bibnamefont {Torlai}}, \bibinfo
  {author} {\bibfnamefont {V.~V.}\ \bibnamefont {Albert}},\ and\ \bibinfo
  {author} {\bibfnamefont {J.}~\bibnamefont {Preskill}},\ }\bibfield  {title}
  {\bibinfo {title} {Provably efficient machine learning for quantum many-body
  problems},\ }\href@noop {} {\bibfield  {journal} {\bibinfo  {journal} {arXiv
  preprint arXiv:2106.12627}\ } (\bibinfo {year} {2021})}\BibitemShut {NoStop}%
\bibitem [{\citenamefont {Huang}\ \emph {et~al.}(2022)\citenamefont {Huang},
  \citenamefont {Broughton}, \citenamefont {Cotler}, \citenamefont {Chen},
  \citenamefont {Li}, \citenamefont {Mohseni}, \citenamefont {Neven},
  \citenamefont {Babbush}, \citenamefont {Kueng}, \citenamefont {Preskill}
  \emph {et~al.}}]{huang2022quantum}%
  \BibitemOpen
  \bibfield  {author} {\bibinfo {author} {\bibfnamefont {H.-Y.}\ \bibnamefont
  {Huang}}, \bibinfo {author} {\bibfnamefont {M.}~\bibnamefont {Broughton}},
  \bibinfo {author} {\bibfnamefont {J.}~\bibnamefont {Cotler}}, \bibinfo
  {author} {\bibfnamefont {S.}~\bibnamefont {Chen}}, \bibinfo {author}
  {\bibfnamefont {J.}~\bibnamefont {Li}}, \bibinfo {author} {\bibfnamefont
  {M.}~\bibnamefont {Mohseni}}, \bibinfo {author} {\bibfnamefont
  {H.}~\bibnamefont {Neven}}, \bibinfo {author} {\bibfnamefont
  {R.}~\bibnamefont {Babbush}}, \bibinfo {author} {\bibfnamefont
  {R.}~\bibnamefont {Kueng}}, \bibinfo {author} {\bibfnamefont
  {J.}~\bibnamefont {Preskill}}, \emph {et~al.},\ }\bibfield  {title} {\bibinfo
  {title} {Quantum advantage in learning from experiments},\ }\href@noop {}
  {\bibfield  {journal} {\bibinfo  {journal} {Science}\ }\textbf {\bibinfo
  {volume} {376}},\ \bibinfo {pages} {1182} (\bibinfo {year}
  {2022})}\BibitemShut {NoStop}%
\bibitem [{\citenamefont {Huba{\v{c}}}\ and\ \citenamefont
  {Wilson}(2010)}]{hubavc2010brillouin}%
  \BibitemOpen
  \bibfield  {author} {\bibinfo {author} {\bibfnamefont {I.}~\bibnamefont
  {Huba{\v{c}}}}\ and\ \bibinfo {author} {\bibfnamefont {S.}~\bibnamefont
  {Wilson}},\ }\bibfield  {title} {\bibinfo {title} {Brillouin-wigner methods
  for many-body systems},\ }in\ \href@noop {} {\emph {\bibinfo {booktitle}
  {Brillouin-Wigner Methods for Many-Body Systems}}}\ (\bibinfo  {publisher}
  {Springer},\ \bibinfo {year} {2010})\ pp.\ \bibinfo {pages}
  {133--189}\BibitemShut {NoStop}%
\bibitem [{\citenamefont {Rayleigh}(1896)}]{rayleigh1896theory}%
  \BibitemOpen
  \bibfield  {author} {\bibinfo {author} {\bibfnamefont {J.~W. S.~B.}\
  \bibnamefont {Rayleigh}},\ }\href@noop {} {\emph {\bibinfo {title} {The
  theory of sound}}},\ Vol.~\bibinfo {volume} {2}\ (\bibinfo  {publisher}
  {Macmillan},\ \bibinfo {year} {1896})\BibitemShut {NoStop}%
\bibitem [{\citenamefont
  {Schr{\"o}dinger}(1926)}]{schrodinger1926quantisierung}%
  \BibitemOpen
  \bibfield  {author} {\bibinfo {author} {\bibfnamefont {E.}~\bibnamefont
  {Schr{\"o}dinger}},\ }\bibfield  {title} {\bibinfo {title} {Quantisierung als
  eigenwertproblem},\ }\href@noop {} {\bibfield  {journal} {\bibinfo  {journal}
  {Annalen der physik}\ }\textbf {\bibinfo {volume} {385}},\ \bibinfo {pages}
  {437} (\bibinfo {year} {1926})}\BibitemShut {NoStop}%
\end{thebibliography}%

\clearpage

\pagebreak

\setcounter{page}{1}
\setcounter{equation}{0}
\setcounter{figure}{0}
\renewcommand{\theequation}{S.\arabic{equation}}
\renewcommand{\thefigure}{S\arabic{figure}}
\renewcommand*{\thepage}{S\arabic{page}}

\onecolumngrid

\begin{center}
{\large \textbf{Supplementary Information for \\``\ourtitle"}}\\
\vspace{0.25cm}

Weishun Zhong,$^{1,2,3}$ Oles Shtanko$^4$, and Ramis Movassagh$^{2,5}$

\vspace{0.2cm}

\textit{
$^1$\mitaff\\
$^2$\cambridge\\
$^3$\iasaff\\
$^4$\almaden\\
$^5$\google}
\end{center}


\vspace{1cm}

In Section~\ref{intro} we present details of the system defined in Eq.~\eqref{eqn:hamiltonian} and the Brillouin-Wigner perturbation series. Our central result in this section is Lemma~\ref{lem1}, which derives our main results. Section~\ref{app:proof_thm1} contains the proof of Theorem~\ref{thm1}, Section~\ref{app:proof_thm2} contains the proof of Theorem~\ref{thm2}, Section~\ref{app:proof_thm3} contains the proof of Theorem~\ref{thm3}, Section~\ref{app:optimal_theorem3} contains the proof of Proposition~\ref{prop1}. In Section~\ref{app:QNN} we present details of the numerical simulations performed in the main text. In particular, we include details on how to construct our model stabilizer codes as well as our Quantum Neural Networks (QNNs).

\section{Perturbed stabilizer codes}
\label{intro}
In this section, we first describe the properties of the ideal stabilizer code, then set up the perturbed problem as well as the corresponding Brillouin-Wigner (BW) series used throughout the proofs. 

Let $S_a$ be a set of independent stabilizers with $a\in\{1,\dots,n-1\}$, then we define stabilizer Hamiltonain $H_0$ as
\be
H_0 = -\sum_{a=1}^{n-1} S_a.
\ee
The eigenstates and eigenenergies of this unperturbed Hamiltonian are
\be \label{app:unperturbed_states}
H_0|\psi_k^{(0)}\> = E_k^{(0)}|\psi_k^{(0)}\>.
\ee
In particular, the groundspace of is doubly degenerate, spanned on a two-dimensional basis of codewords. We define the unperturbed codewords as $\ket{\Omega^{(0)}_i}$, $i=0,1$, satisfying $S_a \ket{\Omega^{(0)}_i} = +\ket{\Omega^{(0)}_i}$ for all $a$. 
 The stabilizer code $\{\ket{\Omega^{(0)}_i}\}$ is said to have code distance $d=2k+1$ if it can be corrected up to $k$-qubit Pauli errors. 
This condition implies that for any $p$-local Pauli operator $P_\alpha$ such that $2p+1\leq d$, the stabilizer code satisfies the Knill-Laflamme (KL) condition,
    \be \label{appeqn:KL condition}
        \bra{\Omega_k^{(0)}} P_\alpha P_\beta \ket{\Omega_{k'}^{(0)}} =  \epsilon^0_{\alpha\beta} \delta_{kk'},
    \ee
where $\epsilon^0_{\alpha\beta}$ is a Hermitian matrix that saisfies $\epsilon^0_{\alpha\alpha} = 1$. Conversely, when $2p+1>d$, we assume that there exist at least one Pauli operator $P_{\rm err}$ of weight $d$ such that

\be\label{eq:logical_error}
\|\delta O\|_1>0, \qquad \delta O := O-\frac 12 \Tr(O) O, \qquad O_{kk'} = \bra{\Omega_k^{(0)}} P_{\rm err} \ket{\Omega_{k'}^{(0)}},
\ee
where $\|\cdot\|_1$ denotes $1$-norm. Eq.~\eqref{eq:logical_error} states that above code distance, there exist at least one error operator that will make the codewords non-orthogonal.

The unperturbed stabilizer code $\ket{\Omega^{(0)}_i}$ also has the following property. 
\begin{clm}
\label{prop:excite}
For any eigenstate $|\psi^{(0)}_k\>$ of the Hamiltonian $H_0$ and any Pauli operator $P$, the state $P|\psi_k^{(0)}\>$ is also an eigenstate of $H_0$.
\end{clm}
\begin{proof}
The eigenstate $\ket{\psi_k^{(0)}}$ has the energy $E^{(0)}_k = -\sum_a s_a$, where $s_a$ are values of the individual stabilizers, $S_a \ket{\psi_k^{(0)}} = s_a \ket{\psi_k^{(0)}}$. We consider the action of the unperturbed Hamiltonian $H_0$ on $P\ket{\psi_k^{(0)}}$, 
\begin{align}
    H_0 P\ket{\psi_k^{(0)}} &= -\sum_a S_a P \ket{\psi_k^{(0)}} = -\sum_a (-1)^{r_{a}} P S_a \ket{\psi_k^{(0)}} \nonumber \\
    &= -  \sum_a (-1)^{r_{a}} s_a P\ket{\psi_k^{(0)}} := \tilde E^{(0)}_{k} P \ket{\psi_k^{(0)}},
\end{align}
where $\tilde E^{(0)}_{k} := -\sum_a (-1)^{r_a+1} s_a$,  $r_a = 0$ if $[P, S_a]=0$ and $r_a = 1$ if $\{P,S_a\}=0$. Therefore, the resulting state $P\ket{\psi_k^{(0)}}$ is still an eigenstate of $H_0$ (possibly the same as $\ket{\psi_k^{(0)}}$).
\end{proof}

Next, we consider a perturbation of $H_0$ by adding a generic $k$-local term $V$, as shown in Eq.~\eqref{eqn:hamiltonian} of the manuscript. The resulting Hamiltonian is
\be\label{eqs:real_ham}
H = H_0+\lambda V,
\ee
where $\lambda$ is a small real parameter. 
We focus on the spectral problem
\be
H|\psi_k\> = E_k|\psi_k\>,
\ee
where $|\psi_k\>$ and $E_k$ are the eigenvectors and eigenenergies of the Hamiltonian $H$. To describe the structure of the code, we need an analytical expression for the two lowest energy eigenstates $|\psi_0\>$ and $|\psi_1\>$.
To find one, we use the Brillouin-Wigner (BW) perturbative expansion \cite{hubavc2010brillouin} of the form
\be \label{BW_vanila}
|\psi_k\>= \frac{1}{\mathcal{N}_k}\sum_{j=0}^\infty \lambda^j|\psi_k^{(j)}\>,
\ee
where $\mathcal{N}_k$ is a normalization coefficient and the correction terms are
\be\label{BW_and_propagator}
|\psi_k^{(j)}\> = (G_k V)^{j}|\psi_k^{(0)}\>, \qquad G_{k} = \sum_{m\neq k}\frac{|\psi_m^{(0)}\>\<\psi_m^{(0)}|}{E_k-E^{(0)}_m}.
\ee
Note that $|\psi_k^{(j)}\>$ for $j>0$ are not normalized and depend on $\lambda$ through $E_k$ on the RHS. The associated energies of the states are
\begin{align}
\label{eqn:En_series}
    E_k &= E_k^{(0)} + \sum_{j=1} \lambda^j \bra{\psi_k^{(0)}} V (G_k V)^{j-1} \ket{\psi_k^{(0)}}.
\end{align}
The advantage of Eq.~\eqref{BW_vanila}-\eqref{BW_and_propagator} 
is that one does not need to consider our case when a degeneracy is present in the unperturbed energy spectrum using a separate formula, 
as required by the standard Rayleigh-Schrödinger (RS) perturbation theory \cite{rayleigh1896theory,schrodinger1926quantisierung}. 

However, with arbitrary choice of the degenerate eigenstates, the perturbation series in Eq.~\eqref{BW_vanila}-\eqref{BW_and_propagator} cannot be treated as a Taylor series in the following sense. If we take into consideration only the first $k$ terms in the series, the contribution from the remaining terms will not have the usual $O(\lambda^{k+1})$ scaling that a converging Taylor series has. This is because the denominator of the propagator $G_k$ can be polynomially small in $\lambda$, thus potentially reducing the power of or even cancelling the $\lambda$-prefactor. 
Use the following Lemma, we show that this problem can be resolved by choosing a proper basis for the degenrate groundspace.
\begin{lem}
\label{lem1} 
There exists a basis $|\phi_{k}^{(0)}\>=\sum_{k'}\beta_{kk'}|\psi_{k'}^{(0)}\>$, where $\beta$ is a $2\times2$ unitary transformation, such that 
\be
\label{app:BW_state1}
|\psi_k\> = \frac{1}{\mathcal N}\sum_{j=0}^{\infty}(\lambda\tilde G_0 V)^{j}|\phi_k^{(0)}\>+O(\lambda^{d}),
\ee
where $\mathcal N>0$ and $\tilde G_0$ takes the form
\be\label{app:BW_state2}
\tilde G_0 = \sum_{m\neq 0,1}\frac{|\psi_m^{(0)}\>\<\psi_m^{(0)}|}{E_0-E^{(0)}_m}.
\ee
\end{lem}
\label{proof_lem1}

The proof is given below at the end of this section. It is important to note that now, in the case of twofold groundspace degeneracy, all denominators in Eq.~\eqref{app:BW_state2} have $O(1)$ scaling due to the presence of a gap in the unperturbed Hamiltonian, see Eq.~\eqref{eqn:En_series}. This allows us to truncate the series with a controlled error. We will use the expression from this Lemma in the following sections.

Another important result is the Knill-Laflamme condition for the resulting perturbed quantum code. With the help of perturbation theory, it can be shown that the generalization of Eq.~\eqref{appeqn:KL condition} to the groundspace codes can be expressed as the following Lemma:
\begin{lem}\label{KLC_lemma}
Let $|\Omega_k\> = \sum_{k'}u_{kk'}|\psi_{k'}\>$ be the superposition of the two lowest eigenstates of the Hamiltonian $H$ in Eq.~\eqref{eqs:real_ham} for some $2 \times 2$ unitary matrix $u$, and let $P_\alpha$ be $p$-local Pauli errors. Then
\begin{equation}\label{eqs:gen_KLC}
\langle \Omega_k | P_\alpha P_\beta | \Omega_{k'} \rangle = \epsilon_{\alpha\beta} \delta_{kk'} +\lambda^{q} h_{\alpha k,\beta k'},
\end{equation}
where $q = \lceil(d-2p)/k\rceil$, $\epsilon_{\alpha\beta}$ and \(h_{\alpha a,\beta b}\) are $O(1)$ matrices satisfying $\epsilon_{\alpha\alpha} = 1$, \(h_{\alpha k,\alpha k'} = 0\) and \(h_{\alpha k, \beta k'} = h^*_{\beta k',\alpha k}\). 
\end{lem}
The proof can be found below in this section.

Finally, we show the following technical result that we will use later. 
\begin{lem}
\label{lem2}
    Upon condition in Eq.~\eqref{eq:logical_error}, there exists $k$-local $V$ and $p$-local Pauli operators $P_\alpha$, $P_\beta$ such that
    \be
    \|T_{\alpha\beta}\|_1>0
    \ee
    where 
    \be\label{eqs?t_tilde}
    T_{\alpha\beta} = t_{\alpha\beta} - \frac 12 \Tr(t_{\alpha\beta})I, \qquad (t_{\alpha\beta})_{kl}: = \sum_{m=0}^{d-2p}\<\Omega_{k}^{(0)}|(V\tilde G_0)^m P_\alpha  P_\beta (\tilde G_0 V)^{d-m-2p}|\Omega_l^{(0)}\>.
    \ee
    Here $I$ is $2\times 2$ identity matrix and $\tilde G_0$ is defined in Eq.~\eqref{app:BW_state2}.
\end{lem}
This Lemma for $p=0$ shows that there exists a perturbation $V$ such that the two lowest eigenstates become non-orthogonal at the $d$-th order of perturbation theory. In the presence of noise affecting $p>0$ qubits, the overlap becomes nonvanishing at the $d-2p$ order.\\

\begin{proof} \textbf{of Lemma~\ref{lem1}.} 
We start from the Brillouin-Wigner series, writing it as
\be\label{eqs:new_bw_series}
\forall k\in\{0,1\}: \quad |\psi_k\> = \frac{1}{\mathcal N_k}M_k|\psi_k^{(0)}\>,\qquad M_k := \sum_{j=0}^\infty (\lambda G_k V)^{j} = \frac{1}{1-\lambda G_kV},
\ee
where $\mathcal N_k$ are normalization coefficients and $G_k$ is the propagator defined in Eq.~\eqref{BW_and_propagator}. Here and below we treat $1$ as the identity operator. Next, for each $k\in\{0,1\}$ we consider a decomposition
\be
G_k = \tilde G_k + \Pi_k, \qquad \tilde G_k = \sum_{m\neq 0,1}\frac{|\psi_m^{(0)}\>\<\psi_m^{(0)}|}{E_k-E^{(0)}_m},\qquad \Pi_k = \frac{|\psi_{1-k}^{(0)}\>\< \psi_{1-k}^{(0)}|}{E_k-E^{(0)}_{1-k}},
\ee
where the operator $\tilde G_k$ does not act on the unperturbed codespace and $\Pi_k$ is a projector to an orthogonal state in the codespace. Using this decomposition, we rewrite
\be\label{eqs:M_exp}
\begin{split}
M_k = \frac{1}{1-\lambda (\tilde G_k+\Pi_k)V} = \tilde M_k\frac{1}{1-\lambda \Pi_k V\tilde M_k},
\end{split}
\ee
where we use the notation
$
\tilde M_k := (1-\lambda \tilde G_kV)^{-1}. 
$
Then, by applying the Taylor series expansion, we can rewrite the second term in the above operator product as
\be
\begin{split}
\frac{1}{1-\lambda \Pi_k V\tilde M_k}  &= \sum_{j=0}^\infty (\lambda \Pi_k V \tilde M_k)^{j} = 1 + \sum_{j=1}^\infty(\lambda \Pi_k V\tilde M_k)^j
\\
&= 1 + \lambda \Biggl(\sum_{j=0}^\infty(\lambda \Pi_k V\tilde M_k)^j\Biggl) \Pi_k V\tilde M_k
= 1 + \lambda\frac{1}{1-\lambda \Pi_k V\tilde M_k} \Pi_k V\tilde M_k
\end{split}
\ee
Here, in the last step, we recombine the series back into the original expression. This gives us an expanded version of Eq.~\eqref{eqs:M_exp} as
\be
M_k = \tilde M_k \Biggl(1+ \lambda \frac{1}{1-\lambda \Pi_k V\tilde M_k} \Pi_k V\tilde M_k\Biggl) =  \tilde M_k \Biggl(1+  \frac{\lambda|\psi_{1-k}^{(0)}\>\<\psi_{1-k}^{(0)}| V\tilde M_k}{E_k-E_{1-k}^{(0)}-\lambda \<\psi_{1-k}^{(0)}| V\tilde M_k|\psi_{1-k}^{(0)}\>} \Biggl).
\ee
The usefulness of this expression becomes apparent when we insert it into Eq.~\eqref{eqs:new_bw_series} and obtain
\be\label{eqs:psi_expansion1}
\forall k\in\{0,1\}: \quad|\psi_k\>  = \frac{1}{\mathcal N_k}\tilde M_k\Bigl(|\psi_k^{(0)}\>+f_k(\lambda)|\psi_{1-k}^{(0)}\>\Bigl),
\ee
where we introduced
\be
f_k(\lambda) = \frac{\lambda\<\psi_{1-k}^{(0)}|V\tilde M_k|\psi_{k}^{(0)}\>}{E_k-E^{(0)}_{1-k}-\lambda \<\psi_{1-k}^{(0)}| V\tilde M_k|\psi_{1-k}^{(0)}\>}, \qquad f_1(\lambda) = -f_0^*(\lambda)+O(\lambda^d)
\ee
Let us rewrite Eq~\eqref{eqs:psi_expansion1} using a new normalization $\tilde{\mathcal{N}}_k$ as
\be
|\psi_k\> = \frac{1}{\tilde{\mathcal N_k}}\sum_{k'}\beta_{kk'}\tilde M_k|\psi_{k'}^{(0)}\>+O(\lambda^d),
\ee
where $\beta$ is a unitary matrix
\be
\beta = \frac1{w}
\begin{pmatrix}
1 & f_0\\
-f^*_0 & 1
\end{pmatrix},
\ee
using $w = \sqrt{1+|f_0|^2}$ and $\tilde N_k = \mathcal N_k/w$.
From the normalization condition we have
\be
\<\psi_k|\psi_l\> = \frac{1}{\tilde{\mathcal N}_k\tilde{\mathcal N}_l}\sum_{k'l'} \beta^*_{kk'}\beta_{ll'}\<\psi^{(0)}_{k'}|\tilde M^\dag_k\tilde M_l|\psi^{(0)}_{l'}\> =\frac{1}{\tilde{\mathcal N}^2_k}\<\psi^{(0)}_{k}|\tilde M^\dag_k\tilde M_k|\psi^{(0)}_{k}\> \delta_{kl}+O(\lambda^d) = \delta_{kl}.
\ee
Therefore
\be
\tilde N_k = [\<\psi^{(0)}_{k}|\tilde M^\dag_k\tilde M_k|\psi^{(0)}_{k}\>]^{-1/2}+O(\lambda^{d}),
\ee
Using the fact that $\tilde M_0 = \tilde M_1+O(\lambda^d)
$, we conclude that $\mathcal N = \tilde{\mathcal N_0} = \tilde{\mathcal N_1}+O(\lambda^d)$, which leads us to
\be
|\psi_k\> = \frac{1}{\mathcal N}\sum_{k'}\beta_{kk'}\sum_{m=0}^\infty (\lambda \tilde G_0 V)^m |\psi_{k'}^{(0)}\>+O(\lambda^{d}).
\ee
This expression concludes our proof.
\end{proof}\\

\begin{proof} \textbf{of Lemma~\ref{KLC_lemma}.} 
We expand both codewords in BW series using Eq.~\eqref{app:BW_state1} to get
\be\label{appqw:sdvi}
    \bra{\Omega_k} P_{\alpha} P_{\beta} \ket{\Omega_{k'}} = \frac{1}{\mathcal{N}^2} \sum_{l,l' \in \{0,1\}} u^*_{kl} u_{k'l'} \sum_{m,m'=0}^{\infty} \lambda^{m+m'} \bra{\Omega_{l}^{(0)}} (V\tilde G_0)^m P_{\alpha} P_{\beta} (\tilde G_0V)^{m'}\ket{\Omega_{l'}^{(0)}}
\ee
Next, we use the decomposition of the perturbation operator $V = \sum_\nu v_\nu K_\nu$ with $k$-local Pauli operators $K_\nu$. Then consider $|\psi^{(0)}_k\>$ to be an eigenstate of the Hamiltonian $H_0$. By Claim~\ref{prop:excite}, $K_\nu |\psi^{(0)}_k\>$ is also an eigenstate of the Hamiltonian $H_0$, hence
\be
\tilde G_0 K_\nu |\psi^{(0)}_k\> = \frac{1}{E_0-E_{\nu k}} K_\nu |\psi^{(0)}_k\>, \qquad E_{\nu k} :=  \<\psi^{(0)}_k|K_\nu H_0 K_\nu |\psi^{(0)}_k\>
\ee
Since $|\Omega_k^{(0)}\> \equiv |\psi_k^{(0)}\> $ for $k=0,1$, there exist real coefficients $\kappa_{\nu_1\dots\nu_m}$ such that
\be
(V\tilde G_0)^m|\Omega_k^{(0)}\> = \sum_{\nu_1\dots\nu_m} \kappa_{\nu_1\dots\nu_m}K_{\nu_1}\dots K_{\nu_m}|\Omega_k^{(0)}\>
\ee
Using this relation and the Knill-Laflamme condition in Eq.~\eqref{appeqn:KL condition}, for all $m+m'< q:=\lceil (d-2p)/k\rceil$ the matrix element in Eq.~\eqref{appqw:sdvi} has the form
\be\label{app:fract920-}
m+m'< q:=\lceil (d-2p)/k\rceil: \qquad \bra{\Omega_{l}^{(0)}} (V\tilde G_0)^m P_{\alpha} P_{\beta} (\tilde G_0V)^{m'}\ket{\Omega_{l'}^{(0)}} = 
\epsilon^{mm'}_{\alpha\beta}\delta_{kk'}.
\ee
We have used the following notations
\be
\epsilon^{mm'}_{\alpha\beta} := \sum_{\mu_1,\dots \mu_m}\sum_{\nu_1,\dots,\nu_{m'}}\kappa^*_{\mu_1,\dots \mu_m}\kappa_{\nu_1,\dots \nu_m}\prod_{i=1}^m v^*_{\mu_i}\prod_{i=1}^{m'} v_{\nu_i}
\ee
Inserting Eq.~\eqref{app:fract920-} into Eq~\eqref{appqw:sdvi}, we get
\be
\begin{split}
    \bra{\Omega_k} P_{\alpha} P_{\beta} \ket{\Omega_{k'}} = \epsilon_{\alpha \beta} \delta_{kk'} + \lambda^{q} h_{\alpha k, \beta k'},
\end{split}
\ee
where
\be
\begin{split}
&\epsilon_{\alpha\beta} = \frac{1}{\mathcal{N}^2}\sum_{m=0}^{q-1}\sum_{m'=0}^{q-1-m} \lambda^{m+m'}\epsilon^{mm'}_{\alpha\beta},\\
&h_{\alpha k,\beta k'} = \frac{1}{\mathcal{N}^2}\sum_{m,m'=0}^{\infty} \lambda^{m+m'-q}H(m+m'-q) \sum_{l,l' \in \{0,1\}} u^*_{kl} u_{k'l'}   \bra{\Omega_{l}^{(0)}} (V\tilde G_0)^m P_{\alpha} P_{\beta} (\tilde G_0V)^{m'}\ket{\Omega_{l'}^{(0)}},
\end{split}
\ee
where $H(\cdot)$ is the Heaviside step function. Note that since $P_{\alpha}^2 = I$, we have $\epsilon_{\alpha \alpha}=1$ and $h_{\alpha k, \alpha k'} = 0$. Also since $\bra{\Omega_{k}} P_{\alpha} P_{\beta} \ket{\Omega_{k'}} = \bra{\Omega_{k'}} P_{\beta} P_{\alpha} \ket{\Omega_{k}}^*$, we have $h_{\alpha k, \beta k'} = h^*_{\beta k', \alpha k}$. 
\end{proof}
\\

\begin{proof} \textbf{of Lemma~\ref{lem2}.} 
According to its definition, the operator $O_{\rm err}$ in Eq.~\eqref{eq:logical_error} is a Pauli operator and thus is a product of single-qubit Pauli operators. Given the freedom in choosing $p$-local operators $P_\alpha$ and $P_\beta$, we can set $P_\alpha P_\beta$ to be a product of $2p$ one-qubit Pauli operators that are already in $O_{\rm err}$. Then we choose a set of \textit{non-overlapping} $k$-local Pauli operators $\{K_i\}_{i=1}^{\lceil (d-2p)/k\rceil}$ to be complementary such that the uncorrectable error operator can be written as
\be
O_{\rm err} = K_1\dots K_{\lceil (d-2p)/k\rceil} P_\alpha P_\beta.
\ee
We also choose operators $K_i$ that have no overlap with $P_\alpha P_\beta$. Now we choose the perturbation $V$ in the form
\be
V = \sum_i K_i.
\ee
For this perturbation, we can rewrite the main object of this Lemma as
\be\label{eqs:overlap}
\begin{split}
(T_{\alpha\beta})_{kl} &:= \sum_{m=0}^{d-2p}\<\Omega_{k}^{(0)}|(V\tilde G_0)^m P_\alpha P_\beta (\tilde G_0 V)^{d-m-2p}|\Omega_l^{(0)}\>\\
&=\sum_{m=0}^{d-2p}\sum_{\vec{a}}\<\Omega_{k}^{(0)}|K_{a_1}\tilde G_0\dots K_{a_m}\tilde G_0 P_\alpha P_\beta \tilde G_0 K_{a_{m+1}}\tilde G_0\dots \tilde G_0K_{a_{d}} |\Omega_l^{(0)}\>,
\end{split}
\ee
 Now, let us define a sequence of energy values
\be
\mathcal E_{\vec{a},m} = \{E_{(a_1),m}E_{(a_1,a_2),m}\dots E_{(a_1,\dots a_d),m}\}
\ee
where
\be
E_{(a_1,\dots, a_l),m} = 
\begin{cases}
\<\Omega_0^{(0)}|K_{a_1}\dots K_{a_l}H_0 K_{a_l}\dots K_{a_1}|\Omega_0^{(0)}\>, \quad \text{if} \quad l<m,\\
\<\Omega_0^{(0)}|K_{a_1}\dots K_{a_l}P_\beta P_\alpha  H_0 P_\alpha P_\beta K_{a_l}\dots K_{a_1}|\Omega_0^{(0)}\>, \quad \text{if} \quad l=m,\\
\<\Omega_0^{(0)}|K_{a_1}\dots K_{a_m}P_\beta P_\alpha K_{a_m}\dots K_{a_l}  H_0 K_{a_l}\dots K_{a_m} P_\alpha P_\beta K_{a_l}\dots K_{a_1}|\Omega_0^{(0)}\>, \quad \text{if} \quad l>m.\\
\end{cases}
\ee
Each $E_{(a_1,\dots, a_l),m}$ is an energy of an eigenstate of the original stabilizer Hamiltonian, therefore $E_{(a_1,\dots, a_l),m}>E_0$. Using the expression for $\tilde G_0$ in Eq.~\eqref{app:BW_state2} we can rewrite Eq.~\eqref{eqs:overlap} as
\be
\begin{split}
(T_{\alpha\beta})_{kl}  = \<\Omega_{k}^{(0)}|O_{\rm err} |\Omega_l^{(0)}\>\sum_{m=0}^{d-2p} \sum_{\vec{a}\in {\rm Perm}(d+1)}\prod_{E\in \mathcal E_{\vec{a},m}}\frac{1}{E_0-E} = c'\<\Omega_{k}^{(0)}|O_{\rm err} |\Omega_l^{(0)}\>,
\end{split}
\ee
where ${\rm Perm}(n)$ is defined as a set of all vectors obtained by permutation of elements in $(1,\dots, n)$ and $c'$ is a non-vanishing coefficient. $|c'|>0$ because it is a sum of terms all having the same sign, as $E_0<E$ for any $E\in \mathcal E_{\vec{a},m}$. Using Eq.~\eqref{eqs:overlap}, we then obtain
\be
\|\tilde T_{\alpha\beta}\|_1 = |c'|\|\delta O\|_1>0.
\ee
This expression concludes our proof.
\end{proof}

\section{Proof of Theorem 1}
\label{app:proof_thm1}

\begingroup
\def\thetheorem{\ref{mythm}}
\addtocounter{thm}{-3}
\begin{thm}
\label{supp_thm1}
\textup{(formal)} Let $H_0$ be a $[[n,1,d]]$ stabilizer Hamiltonian with $r$-local parity check operators $S_a$, $Q_L$ a logical Pauli operator, and $V = \sum_{i=1}^N V_i$, where $V_i$ are $k$-local operators, $k < d$. Then, for any $V$ satisfying $\|[V,S_a]\|>0$ for $N'$ operators $S_a$, the generalization error in Eq.~\eqref{eqn:L2_logical} for noiseless input 
($p=0$) is $\varepsilon_Q(Q_L) = \Theta(\lambda^{4}N'^2/r^4)$. 
\end{thm}
\endgroup

\begin{proof}
Let us rewrite the generalization error measure from Eq.~\eqref{eqn:L2_logical} as
\be
\varepsilon_Q (Q_L) = \int d\theta d\phi \mu(\theta,\phi) \Bigl(\<\Psi_{\theta,\phi}|Q_L|\Psi_{\theta,\phi}\>-\<\Psi^{(0)}_{\theta,\phi}|Q_L|\Psi^{(0)}_{\theta,\phi}\>\Bigl)^2,
\ee
where $Q_L\in\{X_L,Y_L,Z_L\}$ is a logical Pauli operator, $|\Psi_{\theta,\phi}^{(0)}\>$ are logical states of unperturbed stabilizer code and $|\Psi_{\theta,\phi}\>$ are logical states of perturbed code,
\be \label{eq:coeff_thm1}
|\Psi_{\theta,\phi}\> = \sum_{k\in\{0,1\}} \psi_k(\theta,\phi)|\Omega_k\>, \qquad |\Psi^{(0)}_{\theta,\phi}\> = \sum_{k\in\{0,1\}} \psi_k(\theta,\phi)|\Omega^{(0)}_k\>,
\ee
where $\psi_0(\theta,\phi) = \cos \theta$ and $\psi_1(\theta,\phi) = e^{i\phi}\sin\theta$ are two-dimensional wavefunction components. We choose the codewords to be the simultaneous lowest energy eigenstates of the Hamiltonian $H_0$ and eigenstates of the logical operator,
\be
Q_L|\Omega^{(0)}_k\> = (-1)^k|\Omega_k^{(0)}\rangle.
\ee
In turn, the perturbed codewords are a superposition of two lowest eigenstates of the perturbed Hamiltonian $H$,
\be \label{app:alpha_codeword}
|\Omega_k\> = \sum_{k'\in\{0,1\}} \alpha_{kk'}|\psi_{k'}\>,
\ee
where $\alpha_{kk'}$ are components of a two-dimensional unitary rotation. To prove the theorem, we must find the lower bound under the assumption that the choice of $\alpha$ can be arbitrary.

Using these definitions, we rewrite the generalization error as
\be
\varepsilon_Q(Q_L) = \int d\theta d\phi\mu(\theta,\phi)\Biggl(\sum_{k,l\in\{0,1\}}\psi^*_k(\theta,\phi)\psi_l(\theta,\phi)\left[\<\Omega_k|Q_L|\Omega_l\>-\<\Omega^{(0)}_k|Q_L|\Omega^{(0)}_l\>\right]\Biggl)^2.
\ee
 We also use Lemma~\ref{lem1} to rewrite the eigenstates of the Hamiltonian as a series
\be\label{eqs:codew-to-eigv}
|\psi_k\> = \frac{1}{\mathcal N}\sum_{k'\in\{0,1\}} \beta_{kk'} \sum_{m=0}^\infty (\lambda \tilde G_0 V)^m|\Omega_{k'}^{(0)}\>+O(\lambda^{d}),
\ee
where we have chosen $|\psi_k^{(0)}\> = |\Omega_k^{(0)}\>$ for $k\in\{0,1\}$. Combining the expressions in Eqs.~\eqref{app:alpha_codeword} and \eqref{eqs:codew-to-eigv}, we get the resulting expression for the codewords of the perturbed Hamiltonian by the stabilizer codewords as
\be \label{app:pert_code_series}
|\Omega_k\> = \frac{1}{\mathcal N}\sum_{k'} u_{kk'}\sum_{m=0}^\infty (\lambda \tilde G_{0} V)^m|\Omega_{k'}^{(0)}\>+O(\lambda^{d}),
\ee
where we introduce two-dimensional matrix $u = \alpha\beta$, $uu^\dag = I$. With help of this expression, we rewrite
\be
\<\Omega_k|Q_L|\Omega_l\> = \frac{1}{\mathcal N^2}\sum_{k'l'}u^*_{kk'}u_{ll'}\sum_{m,m'=0}^\infty \lambda^{m+m'}\<\Omega_{k'}^{(0)}|(V\tilde G_0)^mQ_L(\tilde G_0V)^{m'}|\Omega_{l'}^{(0)}\>+O(\lambda^{d}).
\ee
Next, considering only lowest orders in perturbation theory, we get the expression
\be\label{eqs:matrix_element_normed}
\<\Omega_k|Q_L|\Omega_l\> = \frac{1}{\mathcal N^2}\sum_{k'l'}u^*_{kk'}u_{ll'}\Bigl((-1)^{k'}\delta_{k'l'}+ \lambda^2 \<\Omega_{k'}^{(0)}|V\tilde G_0 Q_L \tilde G_0V|\Omega_{l'}^{(0)}\>\Bigl)+O(\lambda^{3})
\ee
where we used that $\tilde G_0|\Omega^{(0)}_k\> = 0$. Using perturbation theory and the Knill-Laflamme condition, we can also rewrite the normalization condition as
\be
\<\Omega_k|\Omega_k\> = \frac{1}{\mathcal N^2}\Bigl(1+\lambda^2\<\Omega_0^{(0)}|V\tilde G^2_0V|\Omega_{0}^{(0)}\>+O(\lambda^3)\Bigl) = 1.
\ee
Thus, we conclude that
\be\label{eqs:norm_expansion}
\frac1{\mathcal N^2} = 1-\lambda^2\<\Omega_0^{(0)}|V\tilde G^2_0V|\Omega_{0}^{(0)}\>+O(\lambda^{3}).
\ee
Combining Eqs.~\eqref{eqs:matrix_element_normed} and \eqref{eqs:norm_expansion}, we get the expression
\be
\<\Omega_k|Q_L|\Omega_l\> = \sum_{k'l'}u^*_{kk'}u_{ll'}\Bigl((-1)^{k'}\delta_{k'l'}-\lambda^2 M_{k'l'}\Bigl)+O(\lambda^{3}),
\ee
where we introduced $2\times2$ matrix
\be
M_{kl} = \<\Omega_{k}^{(0)}|V\tilde G_0 (I-Q_L) \tilde G_0V|\Omega_{k}^{(0)}\>\delta_{kl}.
\ee
Since eigenstates in $H_0$ are doubly degenerate at each energy level, it is convenient to split the indices $k$ in $\ket{\psi_k^{(0)}}$ as $\ket{\psi_{i ^\pm}^{(0)}}$ to denote simultaneous eigenstates of $Q_L$ and $H_0$,
\be
Q_L \ket{\psi_{i ^\pm}^{(0)}} = \pm \ket{\psi_{i ^\pm}^{(0)}}, \quad H_0\ket{\psi_{i ^\pm}^{(0)}}=E_{i}^{(0)} \ket{\psi_{i ^\pm}^{(0)}}
\ee
Note that $\ket{\psi_{0+}^{(0)}} \equiv |\Omega_0^{(0)}\>$ and $\ket{\psi_{0-}^{(0)}} \equiv |\Omega_1^{(0)}\>$ by definition. Since $[Q_L, H_0] = 0$, we can present the matrix $M$ as
\be
M  =
\begin{pmatrix}
M_0 & 0\\
0 & M_1
\end{pmatrix}
\ee
where its diagonal elements satisfy
\be
\begin{split}
&M_k = 2 \sum_{i>0}\frac{1}{(E_0-E_{i})^2}|\<\Omega^{(0)}_k|V|\psi_{i^-}^{(0)}\>|^2.
\end{split}
\ee
Now, we can rewrite the generalization error as 
\be
\varepsilon_Q(Q_L) = \int d\theta d\phi \mu(\theta,\phi) \Bigl(\<\theta\phi|u (Z-\lambda^2 M)u^\dag|\theta\phi\>-\<\theta\phi|Z|\theta\phi\>\Bigl)^2+O(\lambda^3),
\ee
where $Z$ is $2\times 2$ Pauli-Z matrix and $|\theta\phi\> = \cos \theta|0\> +e^{i\phi}\sin\theta|1\>$. Assuming that $V = \sum_{i=1}^N V_i$ consists of  local terms that does not commute with at least $N'$ stabilizers $S_a$,  $\<\Omega^{(0)}_k|V|\psi_{i^-}^{(0)}\>$ is non-vanishing for at least $N'$ eigenstates. 
Furthermore, assuming that the majority of local terms do not commute with $\propto r$ stabilizers, the energy difference in the denominator is $E_{i} - E_0 \propto r$. Therefore
\be
\varepsilon_Q(Q_L) = O(\lambda^4 \max(M_0,M_1))=O(\lambda^4 N'^2/r^4).
\ee
This expression concludes our proof.
\end{proof}

\begin{rem}\label{rm1} Condition 
$[H_0,V]=0$ break the notion of generality, both provide vanishing generalization error.
\end{rem}

Indeed, the first condition automatically satisfies $\<\Omega_{k}^{(0)}|V|\psi_{i^-}^{(0)}\> =0$. Using the second condition, one can also show that 
\be
\forall i\neq0: \qquad \<\Omega_{k}^{(0)}|V|\psi_{i^-}^{(0)}\> = \frac{1}{E_{0}^{(0)}-E_{i}^{(0)}}\<\Omega_{k}^{(0)}|[H_0,V]|\psi_{i^-}^{(0)}\> = 0.
\ee
 It is also straightforward to show that these assumptions lead to vanishing higher orders.

\begin{rem}
 For a typical local term $V$, one should expect $N'\sim N$.
\end{rem}

\begin{rem}
 The scaling exponent 4 in Theorem~\ref{supp_thm1} is due to the fact that we use mean square loss to measure generalization error, the exponent would be $2x$ if we used $L_x$ loss function.
\end{rem}

\section{Proof of Theorem 2}
\label{app:proof_thm2}

First, let us introduce the weighted squared expectation for a $2\times 2$ operator $O$ as
\be
\<O\>_{\mu(\theta,\phi)} : = \min_{\theta_0,\phi_0}\int d\theta d\phi\, \mu(\theta_0+\theta,\phi_0+\phi)|\<\psi_{\theta,\phi}|O|\psi_{\theta,\phi}\>|^2,
\ee
where we define $|\psi_{\theta,\phi}\> = \cos \theta |0\>+ e^{i\phi}\sin\theta|1\>$.
For later convenience, we also introduce the following notation for the $\Lambda^{\text{th}}$ order contributions (in the BW expansion) of decoding noisy input states using standard QEC,
\be
\qquad W^{(\Lambda)}_{kl}:=\mathbb E_\alpha \sum_{m=0}^{\Lambda}\Tr(Z_L \mathcal E_{\rm QEC}(|\psi^m_{l\alpha}\>\<\psi^{\Lambda-m}_{k\alpha}|),
\ee
where $|\psi_{k\alpha}^m\> = (V\tilde G_0)^{m}P_\alpha |\Omega_k^{(0)}\>$.
Then we formulate the formal version of the theorem as

\begin{thm}
\label{supp_thm2}
\textup{(Formal)} Let $H_0$ be the stabilizer Hamiltonian for the $[[n,1,d]]$ code, $V = \sum_{i=1}^N V_i$, where $V_i$ are $k$-local operators, $k < d$, and $Q_{\rm QEC}$ is the operator of logical measurement after correction in Eq.~(). Then, for any $V$, input distribution $\mu(\theta,\phi)$, and distribution of $p$-qubit Pauli errors $P_\alpha$, the generalization error in Eq.~\eqref{eqn:L2_logical} is
\be
\varepsilon_Q(Q_{\rm QEC}) = O(\lambda^{2\lceil \xi/k\rceil}),
\ee
where $\xi = d$ if $p=0$ and $\xi = d-2p+1$ if $p>0$. Moreover, for $p\geq 1$ if
\be \label{eq:additional_assumption}
\<W^{(d-2p+1)}\>_{\mu(\theta,\phi)}>0,
\ee
 the generalization error is
 \be 
 \varepsilon_Q(Q_{\rm QEC}) = \Theta(\lambda^{2\lceil(d+1-2p)/k\rceil}).
 \ee
\end{thm}

\begin{rem}
For any given $\mu(\theta,\phi)$, there is no reason to believe that Eq.~\eqref{eq:additional_assumption} is violated for more than a small family of finely tuned perturbations $V$ and noise distributions. Therefore, in the informal version of the theorem, we call $V$ that satisfies this condition to be ``general''.
\end{rem}

\begin{proof}
Consider $S_\alpha$ to be commuting Pauli parity checks for the stabilizer code satisfying $[S_\alpha, S_{\alpha'}] = 0$ and $S_\alpha|\Omega^{(0)}_k\> = |\Omega^{(0)
}_k\>$. For any such code, there exists a set of distinct Pauli corrections ${C_\beta} = C^\dag_\beta$, satisfying $S_\alpha C_\beta = s_{\alpha\beta} C_\beta S_\alpha$, where $s_{\alpha\beta}=\pm1$ are error syndromes. The choice of each operator $C_\beta$ is not unique; we consider the set where each $C_\beta$ has the smallest possible weight. In this representation, these operators must satisfy $2\mathcal S(C_\beta)\leq d-1$, where $\mathcal S(O)$ represents the weight of the operator $O$, i.e. the number of qubits it acts as a non-identity. Using these notations, we can rewrite the error correction as a quantum channel
\be\label{eqs:err_corr_mapx}
\mathcal C(\rho) := \sum_\beta F_\beta \rho F_\beta^\dag, \qquad F_\beta = C_\beta \prod_\alpha \frac 12\left(1+s_{\alpha\beta}S_\alpha\right) =  P C_\beta,
\ee
where $F_\beta$ are Kraus operators for error correction, $P$ is the projector to the codespace. It is convenient to express the expectation values of the logical operators after correction as
\be
\Tr (Q_L\mathcal C(\rho)) = \Tr (Q^{(0)}_L\mathcal C(\rho)) = \Tr (\mathcal C^\dag(Q^{(0)}_L)\rho) = \Tr(\tilde Q_L\rho),
\ee
where $\mathcal C^\dag$ is the adjacent map to the error correcting map $\mathcal C$, and $Q^{(0)}_L = |\Omega^{(0)}_0\>\<\Omega^{(0)}_0|-|\Omega^{(0)}_1\>\<\Omega^{(0)}_1|$ is the restriction of the logical operator $Q_L$ to the codespace, $|\Omega^{(0)}_k\>$ are codewords that, without loss of generality, have been chosen as eigenstates of the logical operator $Q_L$, i.e. $Q_L|\Omega^{(0)}_k\> = (-1)^k|\Omega^{(0)}_k\>$. The operator $\tilde Q_L$ represents an effective logical operator after correction,
\be \label{app:error_corrected_logical}
\tilde Q_L := \mathcal C^\dag(Q^{(0)}_L) \equiv \sum_\mu F_\mu^\dag Q^{(0)}_L F_\mu = \sum_\beta \sum_{r\in\{0,1\}}  (-1)^kC_\beta|\Omega_r^{(0)}\>\<\Omega_r^{(0)}|C_\beta.
\ee

We aim to derive an upper bound on the generalization error
\be \label{eq:gen_error_thm2}
\varepsilon_Q(Q_L) = \mathbb{E}_\alpha\int d\theta d\phi \Bigl(\<\Psi_{\theta,\phi}|P_\alpha\tilde Q_LP_\alpha|\Psi_{\theta,\phi}\>-\<\Psi^{(0)}_{\theta,\phi}|Q_L|\Psi^{(0)}_{\theta,\phi}\>\Bigl)^2
\ee
where $P_\alpha$ are Pauli operators of weight $p$ according to the statement of the theorem.

As a first step, we use the relationship in Eq.~\eqref{app:pert_code_series} between the perturbed and unperturbed codewords to express the following matrix element as
\be \label{eq:order_d}
\<\Omega_k|P_\alpha\tilde Q_L P_\alpha|\Omega_l\> = \frac{1}{\mathcal N^2}\sum_{k'l'}u^*_{kk'}u_{ll'}\sum_{m,m'=0}^\infty \lambda^{m+m'}\<\Omega_{k'}^{(0)}|(V\tilde G_0)^m P_\alpha \tilde Q_L P_\alpha (\tilde G_0V)^{m'}|\Omega_{l'}^{(0)}\>+O(\lambda^{d}).
\ee
Next we note that for a $k$-local correction, we can express $V = \sum_\nu v_\nu K_\nu$, where $K_\nu$ are $k$-local Pauli operators. Then, it is possible to write the following decomposition for the perturbation series,
\be\label{eqs:psi_k}
 (\tilde G_0 V)^m |\Omega_k^{(0)}\> = \sum_{\mu} \kappa_\mu \lambda^{s_\mu}K_\mu |\Omega_k^{(0)}\>, \qquad s_\mu := \mathcal S(K_\mu) \leq m k,
\ee
where $\kappa_\mu$ are certain coefficients that are not divergent in the limit $\lambda\to0$. To prove Eq.~\eqref{eqs:psi_k}, we express 
\be
(\tilde G_0 V)^m|\Omega_k^{(0)}\> = \sum_{\vec a}\frac{v_{a_j}\dots v_{a_1}}{(E_k-E^{(0)}_{f_j})\dots (E_k-E^{(0)}_{f_1})}\lambda^{\mathcal{S}(K_{a_j}...K_{a_1})} K_{a_j}\dots K_{a_1}|\Omega_k^{(0)}\>,
\ee
where $f_i = f_i(\vec a)$ are labels of the energy levels that depend on the sequence $\vec a$.

Using Eq.~\eqref{eqs:psi_k} leads us to the expression
\be\label{eqs:87sdyvs}
\<\Omega_k|P_\alpha\tilde Q_L P_\alpha|\Omega_l\> = \sum_{k'l'}u^*_{kk'}u_{ll'}\sum_{\mu\nu}\kappa^*_\mu\kappa_\nu \lambda^{s_\mu+s_\nu} \sum_\beta\sum_r(-1)^r \<\Omega_{k'}^{(0)}|K_\mu P_\alpha C_\beta|\Omega^{(0)}_r\>\<\Omega^{(0)}_r| C_\beta P_\alpha K_\nu|\Omega_{l'}^{(0)}\>+O(\lambda^{d}).
\ee
For the convenience of computing this expression, let us introduce the following definition. For any two Pauli operators $A$ and $B$, we say that the action of $A$ is isomorphic to action of $B$, or $A\sim B$, if $A|\Omega^{(0)}_{k}\> = B|\Omega^{(0)}_{k}\>$ for $k\in\{0,1\}$. 
Then, for any code stabilized by Pauli operators, we can rewrite
\be
\<\Omega_a^{(0)}|K_\mu P_\alpha C_\beta|\Omega_b^{(0)}\> = e^{i\phi_{\mu\alpha\beta}}\begin{cases}
\delta_{ab}, \quad &\text{if} \quad K_\mu P_\alpha C_\beta  \sim e^{i\phi_{\mu\alpha\beta}}I\\
\delta_{a\overline{b}}, \quad &\text{if} \quad K_\mu P_\alpha C_\beta \sim e^{i\phi_{\mu\alpha\beta}} X_L\\
 (-1)^ai\delta_{a\overline{b}}, \quad &\text{if} \quad K_\mu P_\alpha C_\beta \sim
 e^{i\phi_{\mu\alpha\beta}} Y_L\\
(-1)^a\delta_{ab}, \quad &\text{if} \quad K_\mu P_\alpha C_\beta \sim e^{i\phi_{\mu\alpha\beta}} Z_L\\
0, \qquad &\text{otherwise}
\end{cases}
\ee
where $\overline{a} := 1-a$ and $\phi_{\mu\alpha\beta}$ is a phase.

Keeping these expressions in mind, let us focus on the product of matrix elements in Eq.~\eqref{eqs:87sdyvs}, i.e.
\be
 \<\Omega_{k'}^{(0)}|K_\mu P_\alpha C_\beta|\Omega_r^{(0)}\>\<\Omega_r^{(0)}|C_\beta P_\alpha K_\nu|\Omega_{l'}^{(0)}\>
\ee
and consider three possible scenarios where it is different from zero.\\

\begin{itemize}
\item Scenario 1. Both operators of the matrix elements are isomorphic to an identity transformation, i.e. $K_\mu P_\alpha C_\beta\propto e^{i\phi_{\mu\alpha\beta}}I$ and $K_\nu P_\alpha C_\beta  \propto e^{i\phi_{\nu\alpha\beta}}I$. This case also means that $K_\mu K_\nu \sim e^{i\phi_{\mu\nu}} I$, where $\phi_{\mu\nu} = \phi_{\mu\alpha\beta}-\phi_{\nu\alpha\beta}$. The value of the sum is given by the value of this unit-valued product of non-vanishing matrix elements,
\be
\sum_\beta  \<\Omega_{k'}^{(0)}|K_\mu P_\alpha C_\beta|\Omega_r^{(0)}\>\<\Omega_r^{(0)}|C_\beta P_\alpha K_\nu|\Omega_{l'}^{(0)}\> = e^{i\phi_{\mu\nu}}\delta_{k'r}\delta_{rl'}.
\ee
\item Scenario 2. One of the operators is isomorphic to an identity transformation, and the other is isomorphic to a logical operator. For example, consider $K_\mu P_\alpha C_\beta\sim e^{i\phi_{\mu\alpha\beta}} Q'_L$ and $K_\nu P_\alpha C_\beta\sim e^{i\phi_{\nu\alpha\beta}} I$, where $Q'_L\in\{X_L,Y_L,Z_L\}$. This case also means that $K_\mu K_\nu \sim e^{i\phi_{\mu\nu}} Q'_L$. Th resulting product of the matrix elements, according to the Knill-Laflamme condition, is nonzero if and only if 
\be
\mathcal S(K_\mu)+\mathcal S(K_\nu) \geq d.
\ee
Similar conclusion can be made if $K_\mu P_\alpha C_\beta\sim e^{i\phi_{\mu\alpha\beta}} I$ and $K_\nu P_\alpha C_\beta\sim e^{i\phi_{\nu\alpha\beta}} Q'_L$.\\ 

\item Scenario 3. Both operators are isomorphic to logical operators. In this case $K_\mu P_\alpha C_\beta\sim e^{i\phi_{\mu\alpha\beta}}Q'_L$ and $K_\nu P_\alpha C_\beta\sim e^{i\phi_{\nu\alpha\beta}}Q''_L$, where $Q'_L,Q''_L\in\{X_L,Y_L,Z_L\}$. In this case
\be
\begin{split}
&\mathcal S(K_\mu)+\mathcal S(P_\alpha)+\mathcal S(C_\beta) \geq d,\\
&\mathcal S(K_\nu)+\mathcal S(P_\alpha)+\mathcal S(C_\beta) \geq d.\\
\end{split}
\ee
Summing these two expression we get
\be
\mathcal S(K_\mu)+\mathcal S(K_\nu)\geq 2d-2\mathcal S(C_\beta)-2\mathcal S(P_\alpha)\geq d+1-2p,
\ee
where we have taken into account that any correction operator has weigh never reaches the half of the half-weight of the logical operator, i.e. $2\mathcal S(C_\beta)\leq d-1$.
\end{itemize}

\vspace{0.5cm}

Based on the aforementioned results, we conclude that the elements of the matrix should take the form
\be\label{eqs:OE_expression}
\begin{split}
\<\Omega_k|P_\alpha\tilde Q_L P_\alpha|\Omega_l\> & = \frac{1}{\mathcal N^2}\sum_{k'l'}u^*_{kk'}u_{ll'}\sum_r(-1)^r\delta_{k'r}\delta_{l'r}\sum_{\mu\nu}\lambda^{s_\mu+s_\nu}e^{i\phi_{\mu\nu}}\kappa^*_\mu\kappa_\nu  H(d-2p-s_\mu-s_\nu)\\
&+\lambda^{d-2p+1}\frac{1}{\mathcal N^2} \sum_{k'l'}u^*_{kk'}u_{ll'}\sum_{m=1}^{d-2p+1}\<\Omega^{(0)}_k|(VG_0)^mP_\alpha\tilde Z_L P_\alpha (VG_0)^{d-2p+1-m}|\Omega^{(0)}_l\>\\
&+O(\lambda^{\min(\lceil d/k\rceil,\lceil (d-2p+2)/k\rceil})),
\end{split}
\ee
where $H(x)$ is Heaviside step function defined as $H(x) = 1$ if $x\geq 0$ and $H(x) = 0$ otherwise. Recalling that $\tilde Z_L = \mathcal E^\dag_{QEC} (Z_L)$, it is not hard to see that the second term of this expression includes the matrix $W_{kl}$ defined in Eq.~\eqref{eq:additional_assumption},
\be
\sum_{m=1}^{d-2p+1}\<\Omega^{(0)}_k|(VG_0)^mP_\alpha\tilde Z_L P_\alpha (VG_0)^{d-2p+1-m}|\Omega^{(0)}_l\> = \sum_{m=1}^{d-2p+1}\Tr \left(Z_L \mathcal E_{\rm QEC}(|\psi^{m}_l\>\<\psi^{d-2p+1-m}_k|)\right) := W_{kl}
\ee
Also, from the wavefunction normalization condition, we have
\be \label{eq:norm_thm2}
\<\Omega_k|\Omega_k\> = 1 = \frac{1}{\mathcal N^2}\sum_{\mu\nu}\lambda^{s_\mu+s_\nu}e^{i\phi_{\mu\nu}}\kappa^*_\mu\kappa_\nu  H(d-2p-s_\mu-s_\nu)+O(\lambda^{\lceil d/k\rceil}).
\ee
 As the result, we get
\be
\<\Omega_k|P_\alpha \tilde Q_L P_\alpha |\Omega_l\> = \left(u^\dag (Z+\lambda^{\lceil (d-2p+1)/k\rceil}\frac{1}{\mathcal N^2}W)u\right)_{kl} +O(\lambda^{\min(\lceil d/k\rceil,\lceil (d-2p+1)/k\rceil+1)}).
\ee
This means that, generally, we can choose $\alpha$ and, therefore unitary $u = \alpha\beta$, such that
\be
\sum_\alpha \<\Omega_k|P_\alpha \tilde Q_L P_\alpha |\Omega_l\> = \<\Omega^{(0)}_k| Q_L  |\Omega^{(0)}_l\>+O(\lambda^{\min(\lceil d/k\rceil,\lceil (d-2p+1)/k\rceil}) \equiv \<\Omega^{(0)}_k| Q_L  |\Omega^{(0)}_l\>+O(\lambda^{\lceil \xi/k\rceil}),
\ee
where $\xi = d$ if $p=0$ and $\xi   = d+1-2p$ if $p\geq 1$.

Now for $p\geq 1$, the generalization error Eq.~\eqref{eq:gen_error_thm2} is
\be
\begin{split}
\varepsilon_Q(\tilde Q_L) &:= \int d\theta d\phi\mu(\theta,\phi) \Biggl(\sum_{k,l\in\{0,1\}}\psi^*_k(\theta,\phi)\psi_l(\theta,\phi) \mathbb E_\alpha \left[\<\Omega_k|P_\alpha \tilde Q_L P_\alpha |\Omega_l\>-\<\Omega^{(0)}_k|Q_L|\Omega^{(0)}_l\>\right]\Biggl)^2\\
&= \frac{\lambda^{\lceil 2(d-2p+1)/k\rceil}}{\mathcal{N}^4} \int d\theta d\phi\, \mu(\theta,\phi)|\<\psi_{\theta,\phi}|u^\dag W^{(d-2p+1)} u|\psi_{\theta,\phi}\>|^2 +O(\lambda^{R})\\
&\geq \frac{\lambda^{\lceil 2(d-2p+1)/k\rceil}}{\mathcal{N}^4} \<W^{(d-2p+1)}\>_{\mu(\theta,\phi)}+O(\lambda^{R}),
\end{split}
\ee
where $R = \min(\lceil d/k\rceil,\lceil (d-2p+1)/k\rceil+1)$. Here, we use the fact that $u|\psi_{\theta,\phi}\> = |\psi_{\theta-\theta_0,\phi-\phi_0}\>$ for some $\theta_0\in[-\pi,\pi]$ and $\phi_0\in[-\pi,\pi]$. Using Eq.~\eqref{eq:additional_assumption}, we have
\be
\varepsilon_Q(\tilde Q_L) = \Theta(\lambda^{\lceil 2(d-2p+1)/k\rceil}).
\ee
This expression concludes our proof.

\end{proof}

\section{Proof of Theorem 3}
\label{app:proof_thm3}

We start from the formal version  of the Theorem stated as
\begin{thm}\label{formal_thm_qnn}
\textup{(Formal)} 
Define the quantum neural network observable \(Q_{\text{QNN}} = U^\dag_Q Z_0 U_Q\), where \(U_Q\) is a unitary transformation and \(Z_0\) is the Pauli Z operator on the target qubit. Then for every Pauli observable \(Q\in \{X, Y, Z\}\) there exists a corresponding unitary \(U_Q\) such that the error \(\varepsilon_Q(Q_{\text{QNN}})) = O(\lambda^{4\lceil (d-2p)/k\rceil})\).
\end{thm}

\begin{proof}
The proof begins by noting that the Pauli-$Z$ operator acting on the output qubit, denoted by $Z_0$, has eigenvalues of $\pm 1$. Consequently, the transformation it undergoes can always be expressed as
\be \label{app:basis}
Q_{QNN}:=U_Q^\dag Z_0 U_Q = \sum_{\alpha}\sum_{k\in\{0,1\}}(-1)^k\ket{u_{\alpha k}}\bra{u_{\alpha k}},
\ee
where $|u_{\alpha k}\>$ are orthonormal basis states. There is a one-to-one correspondence between the basis set $\{|u_{\alpha k}\>\}$ and the unitary $U_Q$. Therefore, proving the existence of the basis set is equivalent to proving the existence of $U_Q$.

To show the existence of the basis set, we first show that there exists another complete set of nearly-orthonormal states $\{|v_{\alpha k}\>\}$ such that
\be\label{eqs:another_basis_set}
\begin{split}
&\bra{v_{\alpha k}} v_{\beta k'}\> = \delta_{\alpha\beta}\delta_{kk'}+O(\lambda^{2q}).\\
&\sum_{\alpha l} (-1)^l\bra{\Omega_k} P_\beta \ket{v_{\alpha l}}\bra{v_{\alpha l}} P_\beta \ket{\Omega_{k'}}  = (-1)^k \delta_{kk'}+O(\lambda^{2q})
\end{split}
\ee
where  $q : = \lceil (d-2p)/k\rceil$. Given this set, we can always assign the eigenstates as 
\be
\begin{split}
|u_{\alpha k}\> = |v_{\alpha k}\> + O(\lambda^{2q}),
\end{split}
\ee
where the last  $O(\lambda^{2q})$ term can be decided by the Gram-Schmidt orthonormalization procedure applied to the set of vectors $\{|v_{\alpha k}\>\}$. Finding such states ensures us that
\be\begin{split}
\<\Omega_k|P_\alpha U_Q^\dag Z_0 U_Q P_\beta|\Omega_{k'}\> = \sum_{\alpha l} (-1)^l\bra{\Omega_k} P_\beta \ket{v_{\alpha l}}\bra{v_{\alpha l}} P_\beta \ket{\Omega_{k'}}  +O(\lambda^{2q}) = (-1)^k \delta_{kk'}+O(\lambda^{2q}).
\end{split}
\ee
Then, the generalization error is
\begin{align} \label{app:gen_err_final}
\varepsilon_Q(Q_{\textup{QNN}}) &= \mathbb{E}_{\beta} \int d \theta d \phi \mu(\theta,\phi) \Bigg[\sum_{k l} \psi_k \psi_{l}^* \bigg(\bra{\Omega_k}P_\beta U_Q^\dag Z_0 U_Q P_\beta \ket{\Omega_{l}}-\<\Omega^{(0)}_k| Q_L |\Omega^{(0)}_{l}\>\bigg)\Bigg]^2 = O(\lambda^{4q}),
\end{align}
which would brings us the proof of the theorem.

The rest of the proof is centered around constructing the basis in Eq.~\eqref{eqs:another_basis_set}. We start form Lemma~\ref{KLC_lemma}. We rewrite the result in Eq.~\eqref{eqs:gen_KLC} by ``diagonalizing" the error basis as
\begin{equation}
\langle \Omega_k | B^\dag_\alpha B_\beta | \Omega_{k'} \rangle = e_\alpha\delta_{\alpha\beta} \delta_{kk'} +\lambda^{q} \tilde h_{\alpha k,\beta k'},
\end{equation}
where $B_\alpha = \sum_{\beta}\omega^*_{\alpha\beta}P_\beta$ are new error operators, $\omega_{\alpha\beta}$ is a unitary matrix defining the eigenbasis of $\epsilon_{\alpha\beta}$, $e_\alpha$ are eigenvalues of the matrix $\epsilon_{\alpha\beta}$,
\be
\sum_{\alpha'\beta'} \omega_{\alpha\alpha'} \epsilon_{\alpha'\beta'} \omega^*_{\beta\beta'} = e_\alpha\delta_{\alpha\beta}, \qquad \tilde h_{\alpha k,\beta k'} := \sum_{\alpha'\beta'}\omega_{\alpha\alpha'}h_{\alpha' k,\beta' k'}\omega^*_{\beta\beta'}.
\ee
Assume $D$ is the number of independent errors $\{P_\alpha\}$. Assuming that $e_\alpha>0$ for small enough $\lambda$, our goal is to find a set of (unnormalized) vectors \(|w_{\alpha a}\rangle\) such that states
\be \label{supp:qnn_basis_ansatz}
\begin{split}
&\ket{v_{\alpha k}} = \frac{1}{\sqrt{e_\alpha}}\Bigl(B_{\alpha} \ket{\Omega_k} + \lambda^q\ket{w_{\alpha k}}\Bigl), \qquad \alpha\leq D\\
& \ket{v_{\alpha k}} = |\psi_{\alpha k}\>,\qquad  \alpha >D
\end{split}
\ee
satisfy Eq.~\eqref{eqs:another_basis_set}.
Here $|\psi_{\alpha k}\>$ are states that are orthogonal to the error space spanned on vectors $P_\alpha |\Omega_k\>$. 
To do so, we first expand the l.h.s. of the first condition using the ansatz in Eq.~\eqref{supp:qnn_basis_ansatz}, we rewrite
\be\label{eqs_supp:gsf811}
\begin{split}
\sum_{\alpha, l} (-1)^l\langle\Omega_k| P_\beta |\nu_{\alpha l}\rangle\langle\nu_{\alpha l}| P_\beta |\Omega_{k'}\rangle = \sum_{\alpha, l} \frac{(-1)^l}{e_\alpha}\Bigl(&\langle\Omega_k|P_\beta B_\alpha |\Omega_l\rangle\langle\Omega_l|B_\alpha^\dagger P_\beta |\Omega_{k'}\rangle+\lambda^q\langle\Omega_k|P_\beta B_\alpha |\Omega_l\rangle\langle w_{\alpha l}| P_\beta |\Omega_{k'}\rangle\\
&+\lambda^q\langle\Omega_k|P_\beta |w_{\alpha l}\rangle\langle\Omega_l|B_\alpha^\dagger P_\beta |\Omega_{k'}\rangle+O(\lambda^{2q})\Bigr).
\end{split}
\ee
Next, we use the inverse transformation  that connects Pauli $P_\alpha$ error to the the errors $B_\alpha$ as
\begin{equation}
P_\beta = \sum_{\beta'}\omega_{\beta'\beta}B_{\beta'}.
\end{equation}
The first term in  Eq.~\eqref{eqs_supp:gsf811} becomes
\begin{equation}
\begin{split}
\sum_{\alpha, l} \frac{(-1)^l}{e_\alpha}\langle\Omega_k|P_\beta B_\alpha |\Omega_l\rangle\langle\Omega_l|B_\alpha^\dagger P_\beta |\Omega_{k'}\rangle & = \sum_{\alpha, l} \frac{(-1)^l}{e_\alpha}\sum_{\beta', \beta''}\omega_{\beta''\beta}\omega^*_{\beta'\beta}\langle\Omega_k|B_{\beta'}^\dagger B_\alpha |\Omega_l\rangle\langle\Omega_l|B_\alpha^\dagger B_{\beta''} |\Omega_{k'}\rangle\\
&=\sum_{\beta', \beta''}\omega_{\beta''\beta}\omega^*_{\beta'\beta}\sum_{\alpha, l} \frac{(-1)^l}{e_\alpha}\Bigl(e_\alpha^2\delta_{\alpha\beta'}\delta_{\alpha\beta''}\delta_{kl}\delta_{k'l}+\lambda^q e_\alpha\delta_{\alpha\beta'}\delta_{kl}\tilde h_{\alpha l,\beta'' k'}\\
&\qquad\qquad\qquad\qquad\qquad\qquad+\lambda^q e_\alpha\delta_{\alpha\beta''}\delta_{k'l}\tilde h_{\beta' k,\alpha l}+O(\lambda^{2q})\Bigr)\\
&=(-1)^k\delta_{kk'}\sum_{\beta'}\omega^*_{\beta'\beta}e_{\beta'}\omega_{\beta'\beta}+\lambda^q (-1)^k\sum_{\beta', \beta''}\omega^*_{\beta'\beta}\tilde h_{\beta' k,\beta'' k'}\omega_{\beta''\beta}\\
&\quad +\lambda^q (-1)^{k'}\sum_{\beta', \beta''}\omega^*_{\beta'\beta}\delta_{k'l}\tilde h_{\beta' k,\beta'' k'}\omega_{\beta''\beta}+O(\lambda^{2q})\\
& = (-1)^k\delta_{kk'} \epsilon_{\beta\beta}+\lambda^q\Bigl[(-1)^k+(-1)^{k'}\Bigl]h_{\beta k,\beta k'}+O(\lambda^{2q})  \\
&= (-1)^k\delta_{kk'}+O(\lambda^{2q}),
\end{split}
\end{equation}
where we used the fact that \(h_{\beta k,\beta k'} = 0\) and $\epsilon_{\beta\beta}=1$. The second term in  Eq.~\eqref{eqs_supp:gsf811} becomes
\begin{equation}
\begin{split}
\lambda^q\sum_{\alpha, l}\frac{(-1)^l}{e_\alpha}\langle\Omega_k|P_\beta B_\alpha |\Omega_l\rangle\langle w_{\alpha l}| P_\beta |\Omega_{k'}\rangle & = \lambda^q\sum_{\beta'}\omega^*_{\beta'\beta}\sum_{\alpha, l}\frac{(-1)^l}{e_\alpha}\langle\Omega_k|B^\dag_{\beta'} B_\alpha |\Omega_l\rangle\langle w_{\alpha l}| P_\beta |\Omega_{k'}\rangle \\
& = \lambda^q\sum_{\beta'}\omega^*_{\beta'\beta}\sum_{\alpha, l}(-1)^l \delta_{kl}\delta_{\alpha\beta'}\langle w_{\alpha l}| P_\beta |\Omega_{k'}\rangle+O(\lambda^{2q}) \\
&=\lambda^q(-1)^k\sum_{\beta'}\omega^*_{\beta'\beta} \langle w_{\beta' k}| P_\beta |\Omega_{k'}\rangle +O(\lambda^{2q}).
\end{split}
\end{equation}
Finally, the third term in  Eq.~\eqref{eqs_supp:gsf811} becomes
\begin{equation}
\begin{split}
\lambda^q\sum_{\alpha, l}\frac{(-1)^l}{e_\alpha}\langle\Omega_k|P_\beta |w_{\alpha l}\rangle\langle\Omega_l|B_\alpha^\dagger P_\beta |\Omega_{k'}\rangle & = \lambda^q\sum_{\beta'}\omega_{\beta'\beta}\sum_{\alpha, l}\frac{(-1)^l}{e_\alpha}\langle\Omega_k|P_\beta |w_{\alpha l}\rangle\langle\Omega_l|B_\alpha^\dagger B_\beta |\Omega_{k'}\rangle \\
& = \lambda^q\sum_{\beta'}\omega_{\beta'\beta}\sum_{\alpha, l}(-1)^l \delta_{k'l}\delta_{\alpha\beta'}\langle\Omega_k|P_\beta |w_{\alpha l}\rangle+O(\lambda^{2q}) \\
&=\lambda^q(-1)^{k'}\sum_{\beta'}\omega_{\beta'\beta} \langle\Omega_k|P_\beta |w_{\beta' k'}\rangle +O(\lambda^{2q}).
\end{split}
\end{equation}
Together, these expressions give us
\be\label{eqsupp:cond_one_trans}
\begin{split}
\sum_{\alpha l} (-1)^l\bra{\Omega_k} P_\beta \ket{v_{\alpha l}}\bra{v_{\alpha l}} P_\beta \ket{\Omega_{k'}} & = (-1)^k \delta_{kk'}+\lambda^q\sum_{\beta'}\Bigl((-1)^k\omega_{\beta'\beta}\<\Omega_k|P_\beta|w_{\beta' k'}\>+(-1)^{k'}\omega^*_{\beta'\beta}\<w_{\beta' k}|P_\beta|\Omega_{k'}\>\Bigl)+O(\lambda^{2q})\\
 &= (-1)^k \delta_{kk'}+\lambda^q\Bigl((-1)^k\<\Omega_k|P_\beta|\tilde{w}_{\beta k'}\>+(-1)^{k'}\<\tilde{w}_{\beta k}|P_\beta|\Omega_{k'}\>\Bigl)+O(\lambda^{2q}),
\end{split}
\ee
where we introduced the states
\be\label{eq_supp:basis_change}
|\tilde{w}_{\beta k}\> := \sum_{\beta'}\omega_{\beta'\beta}|w_{\beta' k}\>.
\ee
Similarly, we express the second condition as
\be
\begin{split}
\bra{v_{\alpha k}} v_{\beta k'}\> &= \delta_{kk'}\delta_{\alpha\beta}+\frac{\lambda^q}{\sqrt{e_\alpha e_\beta}}\Bigl(\tilde h_{\alpha k,\beta k'}+\<w_{\alpha k}|B_\beta|\Omega_{k'}\>+\<\Omega_{k}|B^\dag_\alpha|w_{\beta k'}\>\Bigl)+O(\lambda^{2q})\\
&= \delta_{kk'}\delta_{\alpha\beta}+\frac{\lambda^q}{\sqrt{e_\alpha e_\beta}}\sum_{\alpha'\beta'}\omega_{\alpha\alpha'}\omega^*_{\beta\beta'}\Bigl( h_{\alpha' k,\beta' k'}+\<\tilde w_{\alpha' k}|P_{\beta'}|\Omega_{k'}\>+\<\Omega_{k}|P_{\alpha'}|\tilde w_{\beta' k'}\>\Bigl)+O(\lambda^{2q}).
\end{split}
 \ee
Then, to satisfy the conditions in Eq.~\eqref{eqs:another_basis_set}, it is sufficient to have
\be
\begin{split}\label{supp: conditions}
&(-1)^k\<\Omega_k|P_\beta|\tilde{w}_{\beta k'}\>+(-1)^{k'}\<\tilde{w}_{\beta k}|P_\beta|\Omega_{k'}\> = O(\lambda^q),\\
&h_{\alpha k,\beta k'}+\<\tilde w_{\alpha k}|P_\beta|\Omega_{k'}\>+\<\Omega_{k}|P_\alpha|\tilde w_{\beta k'}\> = O(\lambda^q).
\end{split}
\ee
It is straighforward to verify that the following solution satisfies both conditions in Eq.~\eqref{supp: conditions},
\be\label{supp:stabilizer_solution}
|\tilde w_{\alpha k}\> = \sum_{\beta k'}c^*_{\alpha k,\beta k'}P_\beta|\Omega_{k'}\>, \qquad c_{\alpha k,\beta k'}   =  -\frac 12\sum_{\gamma}  h_{\alpha k,\gamma k'}(\epsilon^{-1})_{\gamma \beta}.
\ee
where $(\epsilon^{-1})_{\alpha\beta}$ are the matrix elements of the inverse matrix to $\epsilon_{\alpha\beta}$.
Indeed, if we insert this solution in the first line, we get
\begin{align}
    &(-1)^k\<\Omega_k|P_\beta|\tilde{w}_{\beta k'}\>+(-1)^{k'}\<\tilde{w}_{\beta k}|P_\beta|\Omega_{k'}\> \\
    &\quad = (-1)^k \sum_{\beta' k''}c^*_{\beta k', \beta' k''} \bra{\Omega_{k}} P_{\beta}P_{\beta'}\ket{\Omega_{k''}} + (-1)^{k'} \sum_{\beta' k''} c_{\beta k, \beta' k''}  \bra{\Omega_{k''}} P_{\beta'}P_{\beta}\ket{\Omega_{k'}} \\
    &\quad =(-1)^k \sum_{\beta'}(c_{\beta k', \beta' k} \epsilon_{\beta'\beta})^* + (-1)^{k'}\sum_{\beta'}c_{\beta k, \beta' k'}\epsilon_{\beta' \beta} + O(\lambda^{q})\\
    &\quad =(-1)^k h^*_{\beta k',\beta k} + (-1)^{k'}h_{\beta k, \beta k'} + O(\lambda^{q}) = O(\lambda^q)\\
\end{align}
Similarly, for the next line, we have
\begin{align}
    & h_{\alpha k,\beta k'}+\<\tilde w_{\alpha k}|P_\beta|\Omega_{k'}\>+\<\Omega_{k}|P_\alpha|\tilde w_{\beta k'}\> \\
    &\quad = h_{\alpha k,\beta k'} + \sum_{\beta' k''} \bigg[ c_{\alpha k, \beta' k''} \bra{\Omega_{k''}} P_{\beta'}P_{\beta}\ket{\Omega_{k'}} + c^*_{\beta k', \beta' k''} \bra{\Omega_{k}} P_{\alpha}P_{\beta'}\ket{\Omega_{k''}}  \bigg] \\
    &\quad = h_{\alpha k,\beta k'} + \sum_{\beta'}\bigg[ c_{\alpha k, \beta' k'} \epsilon_{\beta' \beta} + (c_{\beta k', \beta' k} \epsilon_{\beta'\alpha })^*\bigg] + O(\lambda^q)\\
    &\quad = h_{\alpha k,\beta k'} -\frac 12h_{\alpha k,\beta k'} -\frac 12h^*_{
    \beta k',\alpha k}+ O(\lambda^q) = O(\lambda^q)
\end{align}
Then, using Eq.~\eqref{eq_supp:basis_change}, we obtain
\be\label{supp:stabilizer_solution_original}
|w_{\alpha k}\> = -\frac 12 \sum_{\beta \gamma k'}\omega_{\beta \alpha}c^*_{\beta k,\gamma k'}P_\gamma|\Omega_{k'}\>, \qquad c_{\alpha k,\beta k'}   =  -\frac 12\sum_{\gamma}  h_{\alpha k,\gamma k'}(\epsilon^{-1})_{\gamma \beta}.
\ee
Thus, we get the basis for the optimal unitary QNN by substituting Eq.~\eqref{supp:stabilizer_solution} into Eq.~\eqref{supp:qnn_basis_ansatz}. By proving the existence of Eq.~\eqref{eqs:another_basis_set}, we prove Eq.~\eqref{app:gen_err_final} and thus, complete our proof.
\end{proof}

\section{Proof of Proposition~\ref{prop1}}
\label{app:optimal_theorem3}

In the previous scetion, we have proven the existence of the unitary transformation that provides $ \varepsilon_Q(Q_{\rm QNN}) = O(\lambda^{4\lceil (d-2p)/k\rceil})$. In this section, we show that this result is not possible to improve it under generic problem setting.
Let us define the matrix
\be
(T_{\alpha\beta})_{kl} := \sum_{m=1}^{d-2p}\<\Omega^{(0)}_k|(V\tilde G_0)^m P_\alpha P_\beta (V\tilde G_0)^{d-2p-m}|\Omega^{(0)}_l\>.
\ee

\begingroup
\def\thetheorem{\ref{myprop}}
\addtocounter{prop}{-1}
\begin{prop} \label{thm_optimality} Under conditions of Theorem~\ref{formal_thm_qnn}, consider any distribution $\mu(\theta,\phi)$ such that $f(\theta,\phi) = \pm 1$, $\mu(\theta,\phi) = \mu(\theta,\phi+\pi)$, and
\be\label{eqs:theorem3_assumption}
\forall \theta_0,\phi_0: \qquad \mathbb E_\alpha\mathbb E_\beta \int d\theta d\phi\,  \mu(\theta+\theta_0,\phi+\phi_0) |\<\psi_k(\theta,\phi)| T_{\alpha\beta} |\psi_l(\theta,\phi+\pi)\>|^4 > 0.
\ee
Then, for all possible $U_Q$ the generalization error satisfies
 \be
 \varepsilon_Q(Q_{\rm QNN}) = \Omega(\lambda^{4\lceil (d-2p)/k\rceil}).
 \ee
\end{prop}
\endgroup

\begin{lem}\label{lem:binary_state} Any state $|\Psi\>\in \mathcal C := \left\{|\Psi^{(0)}\>,|\Psi^{(0)}\> = \cos\theta |\Omega^{(0)}_0\>+e^{i\phi}\sin \theta |\Omega^{(0)}_1\>, \forall \theta,\phi\in[0,2\pi] \right\}$ and its orthogonal state in the codespace $|\Psi^{(0)}_\perp\>$ have opposite-sign expectation values in respect to $Q_L \in \{X_L, Y_L, Z_L\}$, i.e.
\be
\<\Psi^{(0)}|Q_L|\Psi^{(0)}\> = -\<\Psi^{(0)}_\perp|Q_L|\Psi^{(0)}_\perp\>.
\ee
\end{lem}

\begin{proof} Assuming that
\be
|\Psi^{(0)}\> = \cos\theta |\Omega_0\>+e^{i\phi}\sin \theta |\Omega_1\>, 
\ee
its orthogonal state (up to a phase factor) must take the form
\be
|\Psi^{(0)}_\perp\> = \sin \theta |\Omega_0\>-e^{i\phi}\cos\theta |\Omega_1\>.
\ee

It is straightforward to verify that
\be
\begin{split}
&\<\Psi^{(0)}_\perp|X_L|\Psi^{(0)}_\perp\> =-\sin2\theta\cos\phi  = -\<\Psi^{(0)}|X_L|\Psi^{(0)}\>,\\
&\<\Psi^{(0)}_\perp|Y_L|\Psi^{(0)}_\perp\> =-\sin2\theta\sin\phi  = -\<\Psi^{(0)}|Y_L|\Psi^{(0)}\>,\\
&\<\Psi^{(0)}_\perp|Z_L|\Psi^{(0)}_\perp\> =-\cos 2\theta = -\<\Psi^{(0)}|Z_L|\Psi^{(0)}\>.
\end{split}
\ee
Thus, we arrive at the statement of the Lemma.
\end{proof}\\

\begin{proof} \textbf{of Proposition~\ref{thm_optimality}}. 
The generalization error of decoding with QNN can be rewritten as
\be\label{eq:loss_function_rewrite}
\varepsilon_Q(Q_{\text{QNN}}) =\mathbb E_\alpha  \int d \theta d \phi \mu(\theta,\phi) \Bigl(\Tr( Z_0 U_Q P_\alpha|\Psi\>\<\Psi|P_\alpha U_Q^\dag)-\<\Psi^{(0)}|Q_L|\Psi^{(0)}\>\Bigl)^2.
\ee
For simplicity, let us introduce the following notations for label values
\begin{align} \label{app:condition}
    &f_Q(\theta,\phi): = \bra{\Psi^{(0)}} Q_L \ket{\Psi^{(0)}}, 
    &f_Q^\perp(\theta,\phi) := \bra{\Psi^{(0)}_\perp} Q_L \ket{\Psi^{(0)}_\perp} = -f_Q, 
\end{align}
where we use the results of Lemma~\ref{lem:binary_state}. Note that \(f_Q(\theta,\phi)\) can be \(\pm1\) according to the condition of the Proposition.

Recall that the input states $\ket{\Psi}$ and $\ket{\Psi_{\perp}}$ are constructed from the codewords, with coefficients $\psi_k:= \psi_k(\theta,\phi)$ and $\psi^\perp_{k}:= \psi^\perp_{k}(\theta,\phi)$,
\be
|\Psi\> = \sum_{k} \psi_k|\Omega_k\>, \qquad |\Psi_\perp\> = \sum_{k} \psi^\perp_k|\Omega_k\>,
\ee
where the coefficients satisfy orthogonality condition $\sum_{k} \psi_k \psi^\perp_{k} = 0$.
Then we introduce
\begin{align} 
    &F_\alpha(\theta,\phi): = \Tr( Z_0 U_Q P_\alpha|\Psi\rangle\langle\Psi|P_\alpha U_Q^\dag), 
    &F_\alpha^\perp(\theta,\phi) := \Tr( Z_0 U_Q P_\alpha|\Psi_\perp\rangle\langle\Psi_\perp|P_\alpha U_Q^\dag). 
\end{align} 
Using these notations and note that $f_Q (\theta,\phi)^2 = 1$, we rewrite
\be
\begin{split}\label{eq:loss_function_rewrite_1}
\varepsilon_Q(Q_{\text{QNN}}) & =\mathbb{E}_\alpha  \int d\theta d\phi \mu(\theta,\phi) \left(1-f_Q(\theta,\phi)F_\alpha(\theta,\phi)\right)^2 \\
&=\frac{1}{2}\mathbb{E}_\alpha\mathbb{E}_\beta \left( \int d\theta d\phi \mu(\theta,\phi) \left(1-f_Q(\theta,\phi)F_\alpha(\theta,\phi)\right)^2
+  \int d\theta d\phi \mu(\theta,\phi+\pi) \left(1-f_Q(\theta,\phi)F_\beta(\theta,\phi)\right)^2\right)\\
&=\frac{1}{2}\mathbb{E}_\alpha\mathbb{E}_\beta \int d\theta d\phi \mu(\theta,\phi)\left[ \left(1-f_Q(\theta,\phi)F_\alpha(\theta,\phi)\right)^2
+  \left(1-f^\perp_Q(\theta,\phi)F^\perp_\beta(\theta,\phi)\right)^2\right]\\
&\geq\frac{1}{4}\mathbb{E}_\alpha\mathbb{E}_\beta \int d\theta d\phi \mu(\theta,\phi)\left(2-f_Q(\theta,\phi)F_\alpha(\theta,\phi)-f^\perp_Q(\theta,\phi)F^\perp_\beta(\theta,\phi)\right)^2\\
&=\mathbb{E}_\alpha\mathbb{E}_\beta \int d\theta d\phi \mu(\theta,\phi)\left(1-\frac{1}{2}f_Q(\theta,\phi)\left[F_\alpha(\theta,\phi)-F^\perp_\beta(\theta,\phi)\right]\right)^2,
\end{split}
\ee
where we used the symmetry property of the input state distribution \(\mu(\theta,\phi) =\mu(\theta,\phi+\pi)\) and the inequality \(a^2+b^2\geq \frac{1}{2} (a+b)^2\). Next, let us focus on the quantity
\begin{align}
    \delta F := F_\alpha(\theta,\phi)-F^\perp_\beta(\theta,\phi) \equiv \Tr U_Q^\dag Z_0 U_Q \left(P_\alpha|\Psi\rangle\langle\Psi|P_\alpha- P_\beta|\Psi_\perp\rangle \langle\Psi_\perp|P_\beta \right). 
\end{align}
Since the operator \(U_Q^\dag Z_0 U_Q\) has eigenvalues \(\pm 1\), its operator norm is \(\|U_Q^\dag Z_0 U_Q\| = 1\). Then, by definition of the trace distance \(D(\rho,\sigma)\) between states \(\rho\) and \(\sigma\), we have
\begin{equation}
|\delta F| \leq 2D\left(P_\alpha|\Psi\rangle\langle\Psi|P_\alpha, P_\beta|\Psi_\perp\rangle\langle\Psi_\perp|P_\beta \right).
\end{equation}

Then, the generalized error is bounded as
\begin{align}
    \varepsilon_Q(Q_{\textup{QNN}}) \geq \mathbb E_\alpha\mathbb E_\beta\int d\theta d\phi 
    \mu(\theta,\phi) \Bigl(1-D\bigl(P_\alpha|\Psi\>\<\Psi|P_\alpha, P_\beta|\Psi_\perp\>\<\Psi_\perp|P_\beta \bigl)\Bigl)^2.
\end{align}
For any two pure states $|\Psi_1\>$ and $|\Psi_2\>$, the trace distance satisfies
\be 
D\bigl(|\Psi_1\>\<\Psi_1|,|\Psi_2\>\<\Psi_2|\bigl) \leq \sqrt{1-|\<\Psi_1|\Psi_2\>|^2}\leq 1-\frac 12|\<\Psi_1|\Psi_2\>|^2.
\ee
Using this inequality, we have
\be \label{app:psi_in_basis}
\varepsilon_Q (Q_{\textup{QNN}}) \geq \frac{1}{4} \mathbb E_\alpha\mathbb E_\beta \int d\theta d\phi \mu(\theta,\phi) |\<\Psi|P_\alpha P_\beta|\Psi_\perp\>|^4.
\ee
 Using this representation, we can rewrite
\be \label{app:overlap}
\<\Psi|P_\alpha P_\beta|\Psi_\perp\> = \sum_{kl} \psi^*_k \psi^\perp_{l} \bra{\Omega_k} P_\alpha P_\beta\ket{\Omega_{l}}.
\ee
Now we can expand the codewords in terms of their BW series (Eq.~\eqref{app:pert_code_series}),
\be
\label{app:kn_error_states}
\bra{\Omega_k} P_\alpha P_\beta \ket{\Omega_{l}} = \sum_{k'l'}\frac{u^*_{kk'}u_{ll'}}{\mathcal{N}^2} \sum_{m,m'=0}^\infty\bra{\Omega^{(0)}_{k'}}(\lambda V \tilde{G}_0)^m P_\alpha P_\beta (\lambda\tilde{G}_0 V)^{m'}\ket{\Omega^{(0)}_{l'}} + O(\lambda^{d}).
\ee

In the summation over $i,i'$ in Eq.~\eqref{app:kn_error_states}, for terms with $m+m'< \lceil (d-2p)/k\rceil$, we can apply the Knill-Laflamme condition (Eq.~\eqref{appeqn:KL condition}),
\be
\label{app:kn_sum}
\sum_{\substack{m+m'< \\ \lceil (d-2p)/k\rceil}} \bra{\Omega^{(0)}_{k'}}(\lambda V \tilde{G}_0)^m P_\alpha P_\beta(\lambda\tilde{G}_0 V)^{m'}\ket{\Omega^{(0)}_{l'}} = \sum_{\substack{m+m'< \\ \lceil (d-2p)/k\rceil}}\epsilon^{mm'}_{\alpha\beta}\delta_{k'l'}:= \delta_{k'l'} \tilde{\epsilon}_{\alpha\beta},
\ee
where $\epsilon^{mm'}_{\alpha\beta}$ can be different for different $m,m'$. Recall that $u$ is a unitary matrix, we have
\begin{align}
    \bra{\Omega_k} P_\alpha P_\beta \ket{\Omega_{l}} &= \frac{1}{\mathcal{N}^2}  \left( \tilde{\epsilon}_{\alpha\beta}\delta_{kl} + \lambda^{\lceil (d-2p)/k\rceil} \left(u^\dag T_{\alpha \beta} u \right)_{kl}\right) + O(\lambda^{\lceil (d-2p+1)/k\rceil}).
\end{align}
Thus, for any unitary $u$ we have
\be
\begin{split}
\mathbb E_\alpha\mathbb E_\beta \int d\theta d\phi \mu(\theta,\phi)|\<\Psi|P_\alpha P_\beta|\Psi_\perp\>|^4 &= \mathbb E_\alpha\mathbb E_\beta \int d\theta d\phi\,  \mu(\theta,\phi) |\<\psi_k(\theta,\phi)|u^\dag \tilde T_{\alpha\beta} u |\psi_l(\theta,\phi+\pi)\>|^4 \\
&= \mathbb E_\alpha\mathbb E_\beta \int d\theta d\phi\,  \mu(\theta+\theta_0,\phi+\phi_0) |\<\psi_k(\theta,\phi)| \tilde T_{\alpha\beta} |\psi_l(\theta,\phi+\pi)\>|^4>0.
\end{split}
\ee
The last inequality follows from the assumption in Eq.~\eqref{eqs:theorem3_assumption}.
Therefore, we have 
\begin{align}
     \mathbb E_\alpha\mathbb E_\beta \int d\theta d\phi \mu(\theta,\phi)|\<\Psi|P_\alpha P_\beta|\Psi_\perp\>|^4 = \Theta(\lambda^{4\lceil (d-2p)/k\rceil}).
\end{align}
Here $\Omega(\cdot)$ is a big-Omega notation that show that asymptotically at $\lambda\to 0$ the function is both upper bounded and lower bounded by the funcion in Theta. Insetring this result in Eq.~\eqref{app:psi_in_basis}, we get
\be \label{app:thm3_lower_bound}
\varepsilon_Q (Q_{\textup{QNN}}) = \Omega(\lambda^{4\lceil (d-2p)/k\rceil}).
\ee
Finally, using Theorem~\ref{formal_thm_qnn}, we arrive at the statement of the Proposition. 
\end{proof}

\begin{rem}
The assumption in Eq.~\eqref{eqs:theorem3_assumption} serves as a generality condition for the chosen input states and perturbation $V$. As we had shown in Lemma~\ref{lem2}, the matrix $T_{\alpha\beta}$ is not identity, at least for some values of $V$ and noise operators. Therefore, the l.h.s. of Eq.~\eqref{eqs:theorem3_assumption} vanishes in this case only for a fine-tuned choice of input states. More generally, there is no particular reason to believe that this expression vanishes for generic parameters of the problem.
\end{rem}

\begin{rem}
    For noiseless input $(p=0)$, the QNN considered in Theorem~\ref{formal_thm_qnn} can decode input states with arbitrary accuracy, provided that there is no constraint on its architecture.
\end{rem}

\section{Details of numerical simulations}
\label{app:QNN}

In this section, we will delve into the details of the numerical results that are presented in the main text. First, we introduce the stabilizers of the unperturbed codes used in the numerical simulations, as well as the possible errors considered in these codes. Then we outline the numerical procedures for selecting perturbations and obtaining codewords from the perturbed Hamiltonian, as described by Eq.~\eqref{eqn:hamiltonian}. For Fig.~\ref{fig:QEC_panel}, we provide insight into the construction of the ``error-corrected'' logical operators, as expressed in Eq.~\eqref{eqn:logical_tilde}, and elaborate on how we obtain the expectation values. Moving on to Fig.~\ref{fig:QNN_panel}, we outline the steps for training the QNN as specified in Eq.~\eqref{eqn:log-network}-\eqref{eqn:Xbar}, and provide details on how to estimate the generalization error in QNNs. We conclude this section by discussing the effect of QNN depth on the results and by considering possible alternative architectures.

\subsection{Unperturbed codes}
In this subsection, we list the stabilizers $S_a$ for the unperturbed stabilizer codes used in the main text, as well possibe Pauli errors that can occur in these codes. 

\subsubsection{[[5,1,3]] code}
The stabilizers of this code are \cite{laflamme1996perfect}
\begin{align*}
    S_1 &= XZZXI, \\
    S_2 &= IXZZX, \\
    S_3 &= XIXZZ, \\
    S_4 &= ZXIXZ.
\end{align*}
There are $3\times 5 =15$ single-qubit errors, which correspond to a single-Pauli error, $X, Y$ or $Z$ on one of the 5 qubits.

\subsubsection{[[7,1,3]] (Steane) code}
The stabilizers of this code are \cite{steane1996error}
\begin{align*}
    S_1 &= IIIXXXX, \\
    S_2 &= IXXIIXX, \\
    S_3 &= XIXIXIX, \\
    S_4 &= IIIZZZZ, \\
    S_5 &= IZZIIZZ, \\
    S_6 &= ZIZIZIZ.
\end{align*}
There are $3 \times 7 = 21$ single-qubit errors, which correspond to a single-Pauli error, $X, Y$ or $Z$ on one of the 7 qubits.

\subsubsection{[[9,1,3]] (Shor) code}
The stabilizers of this code are \cite{shor1995scheme}
\begin{align*}
    S_1 &= ZZIIIIIII, \\
    S_2 &= IZZIIIIII, \\
    S_3 &= IIIZZIIII, \\
    S_4 &= IIIIZZIII, \\
    S_5 &= IIIIIIZZI, \\
    S_6 &= IIIIIIIZZ, \\
    S_7 &= XXXXXXIII, \\
    S_8 &= IIIXXXXXX. 
\end{align*}
There are $3 \times 9 = 27$ single-qubit errors, which correspond to a single-Pauli error, $X, Y$ or $Z$ on one of the 9 qubits.

\subsubsection{[[11,1,5]] code}
The stabilizers are \cite{Grassl:codetables}
\begin{align*}
    S_1 &= XZIZIXIZZII, \\
    S_2 &= IYIZZYIIZZI, \\
    S_3 &= IZXIZXIIIZZ, \\
    S_4 &= IZZYIYZIIZI, \\
    S_5 &= IIZZXXZZIZZ, \\
    S_6 &= IIZZIIYZZIY, \\
    S_7 &= IZIZZZZYIZY, \\
    S_8 &= IIIZIZIZXZX, \\
    S_9 &= IZIZIIZIZXX, \\
    S_{10} &= ZZZZZZIIIII.
\end{align*}
There are $11 \times 3 = 33$ single-qubit errors, which correspond to a single-Pauli error, $X, Y$ or $Z$ on one of the 11 qubits, and $3^2 \times \binom{11}{2} = 495$ two-qubit errors, which correspond to two-Pauli errors on two of the 11 qubits. 

\subsection{QEC simulations}
Given codewords constructed as linear combination of the two lowest fully-perturbed eigenstates $\ket{\psi_k}$,
\be    
\label{app:omega_psi}
\ket{\Omega_k} = \sum_{k'\in \{0,1\}} \alpha_{kk'} \ket{\psi_k'},
\ee
where $\alpha_{kk'}$ is a unitary matrix that we can choose to optimize standard QEC performance. We use the following parameterization:
\be
\label{app:alpha_matrix}
\alpha  =
\begin{pmatrix}
e^{it_0}\cos{t_1} & e^{i(t_0+t_2)}\sin{t_1}\\
e^{i(t_0+t_3)} \sin(t_1) & -e^{i(t_0+t_2+t_3)}\cos(t_1)
\end{pmatrix}.
\ee
Ignoring the global phase, we set $t_0 = 0$. We then choose $t_1$ and $t_2$ by maximizing $\bra{\Omega_0}Z_L\ket{\Omega_0}$, and $t_3$ by maximizing $\bra{\Omega_+} X_L \ket{\Omega_+}$, where $\ket{\Omega_\pm} := (\ket{\Omega_0} \pm \ket{\Omega_1})/{\sqrt{2}}$. \\

Once $\alpha$ is fixed, we build $P$ samples
\be 
\label{app:psi_sample}
\left\{\ket{\Psi_\mu} = \cos(\theta_\mu)\ket{\Omega_0} + \sin(\theta_\mu)e^{i\phi_\mu}\ket{\Omega_1}\right\}_{\mu=1}^P,
\ee
where $\theta_\mu \in [0,\pi]$ and $\phi_\mu\in[0,2\pi]$ are chosen i.i.d. from corresponding uniform distributions. This provides an estimate for the integral with the constant measure $\mu(\theta,\phi) = 1/2\pi^2$.

\subsubsection{Details for Fig.~\ref{fig:QEC_panel}}

In all cases, we estimate the decoding loss $\varepsilon_Q$ on $S=10^3$ samples of the input state $\ket{\Psi_\mu}$ (Eq.~\eqref{app:psi_sample}). The results are then averaged over $10^5$ different realizations of the single-qubit perturbation, $V = \sum_{i=1}^n V_i$, where random correction $V_i$ is acting on qubit $i$, sampled from $2\times2$ GUE, and normalized such that $||V_i||=1$. 

In Fig.~\ref{fig:QEC_panel}(a), after specifying $f_Q(\theta,\phi)$ for each $\ket{\Psi}$ as defined in Eq.~\eqref{eqn:psi} (e.g., $f_X(\theta,\phi) = \sin 2\theta \cos\phi$, $f_Y(\theta,\phi)= \sin 2\theta \sin \phi$, $f_Z(\theta,\phi)=\cos 2\theta)$, we get the exact expectation of the ideal code logical operator by evaluating $\bra{\Psi}Q_L\ket{\Psi}$, where $Q_L \in \{X_L, Y_L, Z_L\}$.

In Fig.~\ref{fig:QEC_panel}(b), we numerically obtain the expectation value of effective error-corrected logical operator $\tilde{Q}_L$, where 
\be
\tilde{Q}_L := \sum_\mu F_\mu^{\dagger} Q_L F_\mu, \qquad F_\mu = C_\mu \prod_\alpha \frac 12\left(1+s_{\alpha\mu}S_\alpha\right) =  P C_\mu,
\ee
where $s_{\alpha\mu}=\pm1$ represents the $\mu$-th error syndrome obtained from measuring projection onto the $a$-th stabilizer subspace,
$C_{\mu}$ represents the single-qubit (5, 7, 9, 11-qubit codes) or two-qubit (11-qubit code) error-correction corresponding to syndrome $s_{\alpha\mu}$, and $P=\sum_{k=0,1}\ket{\Omega^{(0)}_k}\bra{\Omega^{(0)}_k}$ is the projector to the codespace. 

In Fig.~\ref{fig:QEC_panel}(c), for each sample state $\ket{\Psi_\mu}$ (Eq.~\eqref{app:psi_sample}), we first randomly apply an error (single-qubit Pauli error for 5, 7, or 9-qubit code, single or two-qubit Pauli error for 11-qubit code) to it, then obtain the expectation value of $\tilde{Q}_L$.

\subsection{QNN-QEC simulations}

We train the QNN with $N_{T} \sim 10^3$ training sample states as defined in Eq.~\eqref{app:psi_sample}. In some numerical simulations (see below) we use exact states, in the other we apply to some of the states single- or two-qubit Pauli errors. After training, we test our QNN on $N_{V} = 10\times N_T$ samples withheld from the training but independently drawn from the same distribution. The factor of 10 is chosen to ensure that validation set provides statistically significant results. If the error is present in the training set, the states in the validation set are chosen to contain errors as well. For each validation sample state $\ket{\Psi_{\mu \alpha}} := P_\alpha \ket{\Psi_\mu}$, we obtain the expectation value of measuring $Z$ operator at the output qubit, by computing 
\be
\langle Q_{\text{QNN}} \rangle_{\mu\alpha} := \bra{\Psi_{\mu \alpha}} U^\dag_Q Z_0 U_Q \ket{\Psi_{\mu\alpha}}.
\ee
Then we use the result to estimate the generalization error by performing average over the $N_V$ validation sample states,
\be
\varepsilon_Q (Q_{\text{QNN}}) = \frac{1}{MN_V} \sum_{\alpha=1}^{M} \sum_{\mu=1}^{N_V/M} \left( \langle Q_{\text{QNN}} \rangle_{\mu\alpha} - \bra{\Psi^{(0)}_\mu} Q_L \ket{\Psi^{(0)}_\mu} \right)^2.
\ee

\subsubsection{Details for Fig.~\ref{fig:QNN_panel}}
In this section, we present details of the QNN simulation in the main text Fig.\ref{fig:QNN_panel}. Specializing the general architecture in Fig.3(a) to 5-qubit code, we use the 5-qubit QNN architecture in Fig.~\ref{appfig:QNN_5qubit}. In Fig.\ref{fig:QNN_panel}(b)-(d), we have $N_T=1600$ training examples. After training, as described in the previous section, we validate on $N_V$ different validation samples, independently drawn from the same distribution. 

Each two-qubit unitary $U_k$ (Eq.~\eqref{eqn:log-network}) contains $4^2-1=15$ independent parameters. We fix the overall non-measurable angle by choosing $\phi_{00}^{k}=1$ and allow the other 15 angles $\phi_{\alpha \beta}^k$ to be free parameters to be optimized. For the QNN shown in Fig.~\ref{appfig:QNN_5qubit}, there are a total of $(4C+1)\times 15$ parameters. For the result presented in Fig.~\ref{fig:QNN_panel}, we chose $d_C=4$, which corresponds to $255$ parameters. We optimize these parameters with BFGS method using the SciPy \texttt{optimize} library
(maximum number of iterations is $10^5$, tolerance for termination is $10^{-15}$). Note that it is more common to optimize QNN with Stochastic Gradient Descent (SGD), where the gradients are estimated by methods such as finite difference. We find that, however, optimizing with SGD tends to yield larger decoding error, incompatible with the accuracy required by the comparison with perturbation theory. 

In Fig.\ref{fig:QNN_panel}(b)-(c), all the $N_T$ training samples and $N_V$ validation samples are unaffected by errors. In Fig.\ref{fig:QNN_panel}(d), all the samples are noisy: we randomly apply one of the $M=16$ possible single-qubit Pauli errors (including the identity) to each of the sample states in both training and validation set. 

\subsubsection{The effect of QNN depth}
In this section, we discuss the rationale for choosing the architecture presented in Fig.\ref{fig:QNN_panel}(a), as well as the effect of different depths $d_C$ for the QNN. The motivation for calling $C_Q$ the `error-correction circuit' is as follows. We find that for a shallow network, e.g., $d_C=1$, even for ideal stabilizer code ($\lambda=0$), the QNN fails to learn noisy samples (Fig.~\ref{appfig:QNN_5qubit}(b)). However, as the depth $d_C$ increases, the QNN generalization erorr decreases and $\langle X_{\text{QNN}} \rangle$ gradually approaches $f_X$ (Fig.~\ref{appfig:QNN_5qubit}(c)-(d)), and there appears to be a critical depth where the scatter plot between $\langle X_{\text{QNN}} \rangle$ and $f_X$ suddenly falls on the diagonal (QNN expectation value matches the expected measurement value $f_X$). We observe that at $d_C=4$ the scatter plot falls on the diagonal ($\langle X_{\text{QNN}} \rangle \approx f_X$) (Fig.~\ref{appfig:QNN_5qubit}(e)). Therefore, $C_Q$ has the spirit of performing error-correction on the noisy samples, but we emphasis that $C_Q$ is not really performing the error-correction equivalent to the standarad QEC. For $d_C=4$, the typical optimization time is around one hour on a intel i7 @ 1.90GHz 4 Cores CPU (without parallelization). 

\begin{figure*}[t!]
    \centering
        \includegraphics[width=1.\textwidth]{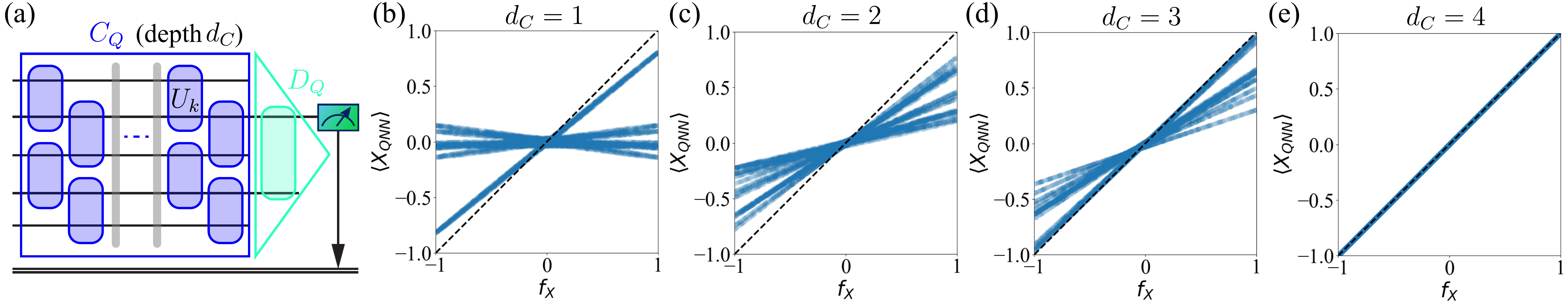}
        \caption{\textbf{5-qubit QNN.} (a) The `error-correction circuit' $C_Q$ has depth $d_C$. The `decoding circuit' $D_Q$ consitst of a single 2-qubit unitary. (b)-(e) The decoding performance improves as depth $d_C$ increases.}
\label{appfig:QNN_5qubit}
\end{figure*}
Next, we compare the generalization error across different circuit depths $d_C$ for the architecture in Fig.~\ref{appfig:QNN_5qubit}. We also present a potentially different architecture with less number of parameters than the one in Fig.\ref{fig:QNN_panel}(a). We consider the same task as in Fig.\ref{fig:QNN_panel}(c) (noiseless inputs). In Fig.~\ref{appfig:QNN_archi}, we consider different depths $d_C=2,3,4$, and find the generalization performance improves as depth increases. Next, we consider a Translationally Invariant (Trans. Inv.) architecture where all the 2-qubit unitaries $U_k$ within the same layer of $C_Q$ (same layer of circuit, see the grey blockcade in Fig.~\ref{appfig:QNN_5qubit}, which includes two layers of qubits) share the same parameters $\phi_{\alpha \beta}^k$. Note that this architecture contains less parameters, e.g., $d_C=4$ (Trans. Inv.      ) has 75 parameters, less than the 255 parameters in $d_C=4$. We find that the generalization performance of this architecture is worse than that of the $d_C=4$ architecture we used in Fig.~\ref{appfig:QNN_5qubit} (grey curve in Fig.~\ref{appfig:QNN_archi}). It is worth noticing that all the different depths and architectures considered in Fig.~\ref{appfig:QNN_5qubit} has roughly the same constant scaling with $\lambda$, suggesting a robust scaling behavior regardless of the specific architecture used. 
\begin{figure*}[t!]
    \centering
        \includegraphics[width=0.4\textwidth]{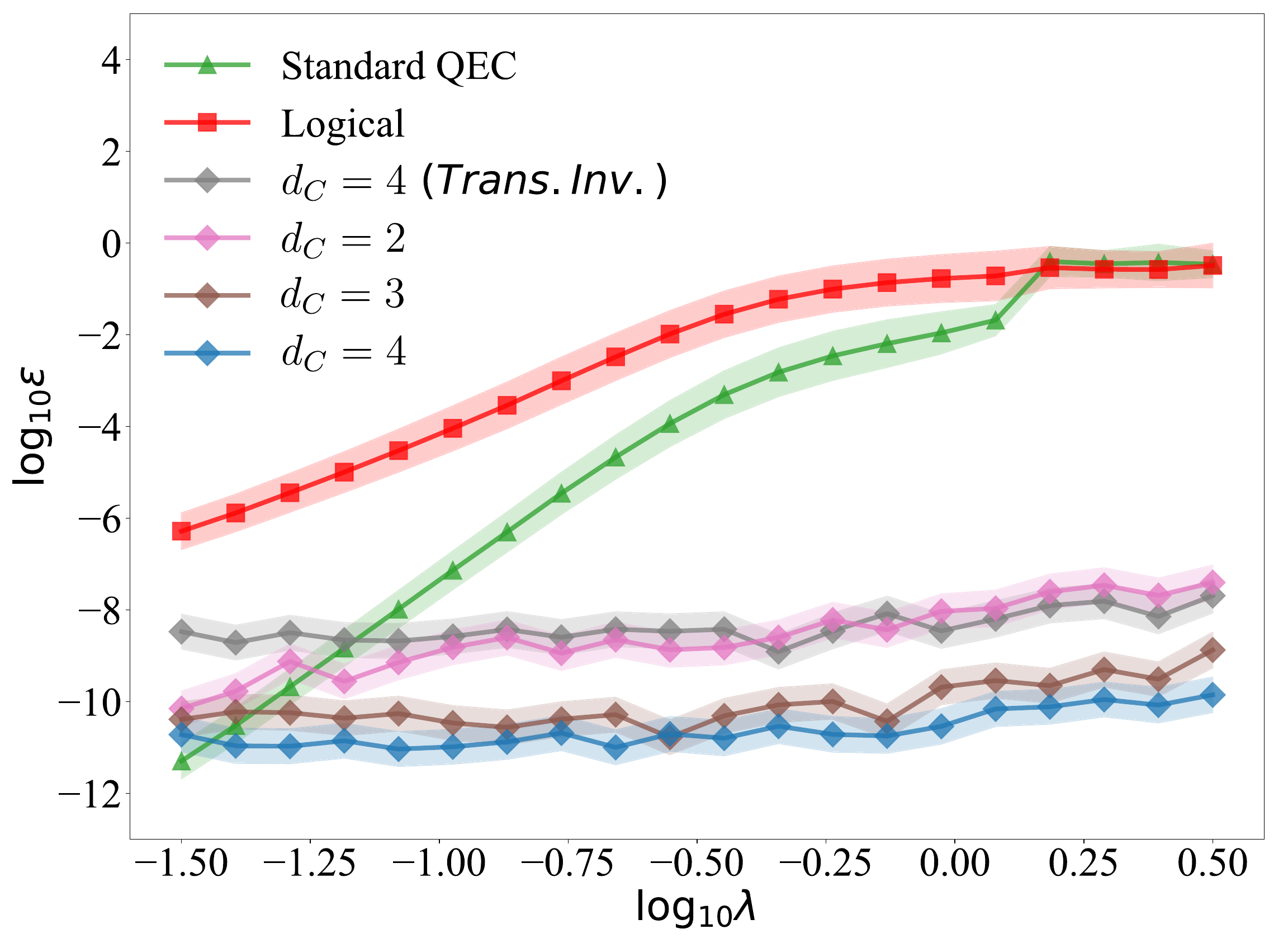}
        \caption{\textbf{Different QNN architectures.} We consider the same task as in Fig.\ref{fig:QNN_panel}(c). $d_C=2,3,4$ corresponds to independent unitaries $U_k$ in Fig.~\ref{appfig:QNN_5qubit}, $d_C=4$ (Trans. Inv.) correspond to same $U_k$ within the same layer. For reference, we reproduce the decoding performance of the standard decoding $\varepsilon_X(Q_{QEC})$ (green curve) and $\varepsilon_X(Q_L)$ (blue curve) as in Fig.\ref{fig:QNN_panel}(c).}
\label{appfig:QNN_archi}
\end{figure*}

\end{document}